\documentclass[10pt]{amsart}
\usepackage{graphicx}
\usepackage[width=430pt,height=620pt,centering]{geometry}
\usepackage[latin1]{inputenc}
\usepackage{enumitem}
\usepackage{amssymb}
\usepackage{xspace}
\usepackage{verbatim}
\usepackage{tikz}
\usetikzlibrary{positioning}
\usetikzlibrary{fit}

\newtheorem{thm}{Theorem}[section]

\newtheorem{prop}[thm]{Proposition}
\theoremstyle{definition}

\theoremstyle{remark}

\numberwithin{equation}{section}

\newcommand{\Harm}[1]{{{\mathcal H}_{#1}}}

\newcommand{\NtSet}{\mathcal{N}}
\newcommand{\Nt}{N}
\newcommand{\T}{t}
\newcommand{\Voc}{\Sigma}
\newcommand{\Axiom}{\mathcal{S}}
\newcommand{\ProdRules}{\mathcal{P}}
\newcommand{\Lang}{\mathcal{L}}

\newcommand{\BigO}[1]{\mathcal{O}(#1)}
\newcommand{\Def}[1]{{\bf #1}}

\newcommand{\Production}{\rightarrow}

\newcommand{\Prob}[1]{\mathbb{P}(#1)}

\newcommand{\WGram}[1]{\mathcal{G}_{#1}}

\newcommand{\Est}{\mathcal{U}}

\newcommand{\W}{\pi}
\newcommand{\WF}[1]{\W(#1)}
\newcommand{\CW}{W}
\newcommand{\D}{\Delta}

\newcommand{\NCW}{m}
\newcommand{\WGF}[2]{{#1}_{#2}(z)}

\newcommand{\PF}[2]{{\mu_{{#1},{#2}}}}
\newcommand{\WProb}[1]{w}

\newcommand{\Mom}{\alpha}
\newcommand{\PMax}[2]{p_{#1,#2}^{\triangle}}

\newcommand{\WMax}[2]{\CW_{#1,#2}^{\triangle}}
\newcommand{\WMin}[2]{\CW_{#1,#2}^{\nabla}}

\newcommand{\UB}{\bullet}
\newcommand{\OP}{\text{{\bf (}}\xspace}
\newcommand{\CP}{\text{{\bf )}}\xspace}

\newcommand{\MCount}[1]{{m_{#1}}}
\newcommand{\MBCount}[2]{{m_{#1,#2}}}
\newcommand{\MPF}[2]{\mu_{#1,#2}}

\newcommand{\BBigO}[1]{\mathcal{O}\left(#1\right)}

\newcommand{\RNAW}{w}
\newcommand{\RNASG}[2]{S_{#1,#2}}

\begin{document}

\title[Collisions in random generation of weighted context-free languages]{Weighted random generation of context-free languages: Analysis of collisions in random urn occupancy models}%
\author{Danièle Gardy}
\address{Laboratoire PRiSM.CNRS UMR 8144 and Université de Versailles St-Quentin en Yvelines, 45 Av. des États-Unis, 78035 Versailles, France}
\email{Daniele.Gardy@prism.uvsq.fr}%

\author{Yann Ponty}
\address{Laboratoire d'Informatique de l'École Polytechnique (LIX),
CNRS UMR 7161/AMIB INRIA. École Polytechnique, 91128 Palaiseau, France}
\email{yann.ponty@lix.polytechnique.fr}%

\date{\today}


%\date{}%
%\dedicatory{}%
%\commby{}%
% ----------------------------------------------------------------
\begin{abstract}
  The present work analyzes the redundancy of sets of combinatorial objects produced by a weighted random generation algorithm proposed by Denise~\emph{et al}.
  This scheme associates weights to the terminals symbols of a weighted context-free grammar, extends this weight definition multiplicatively on words,
  and draws words of length $n$ with probability proportional their weight.
  We investigate the level of redundancy within a sample of $k$ word, the proportion of the total probability covered by $k$ words (coverage),
  the time (number of generations) of the first collision, and the time of the full collection. For these four questions, we use an analytic urn
  analogy to derive asymptotic estimates and/or polynomially computable exact forms. We illustrate these tools by an analysis of an RNA secondary
  structure statistical sampling algorithm introduced by Ding~\emph{et al}.
\end{abstract}
\maketitle

%\tableofcontents
% ----------------------------------------------------------------
\section{Introduction}

%\subsection{Random generation}
The random generation of combinatorial objects is both motivated by the exploration
of complex objects, the empirical assessment of statistical properties and by its
applications to numerous fields (analysis of data structures and algorithms~\cite{Bassino2009},
software testing~\cite{Den06,Canou2009}, bioinformatics~\cite{DiLa03}\ldots). Many approaches have been developed to address
the uniform random generation of combinatorial objects of a given size. Historically, the recursive method, formalized by Wilf~\cite{wilf77},
starts by efficiently pre-computing the numbers of objects accessible from local choices, and uses these numbers during the generation
to perform an uniform random generation as an unbiased walk. This approach was later extended and made fully automatic by Flajolet~\emph{et al}~\cite{flajoletcalculus}
for all decomposable combinatorial classes, i. e. classes that are specified constructively within the symbolic framework
as opposed to implicitly defined by a required property.
Finally Duchon~\emph{et al}~\cite{fullboltz} recently relaxed this scheme through Boltzmann sampling.

%\subsection{Weighted context-free grammars}
  Yet certain contexts require a non-uniform -- yet controlled -- distribution to be captured, giving rise to various approaches~\cite{Brlek2006} for the non-uniform generation.
  Denise~\emph{et al}~\cite{Denise2000} introduced weighted context-free grammars where a weight function, defined on the terminals and extended multiplicatively on words,
  induces a Boltzmann distribution over each subset of words of a given length $n$. The resulting languages are then used
  as models for objects following non-uniform distributions, of which natural instances can be found in bioinformatics~\cite{PoTeDe06}.
  An adaptation of the recursive method was proposed~\cite{Denise2000} to draw words of a given size $n$ with respect to a
  weighted distribution. Multidimensional Boltzmann versions of the weighted samplers were also proposed
  for weighted languages by Bodini~\emph{et al}~\cite{Bodini2010}.

   However weighted distributions, by assigning probabilities to possible words that scale exponentially within a class of size,
   may induce a -- possibly large -- redundancy within sampled set of words.
   Since the probability of a word is exactly and efficiently computable such a redundancy is not informative and should be avoided.
   Furthermore, if a non-redundant sample of given cardinality $k$ is expected, one may find situations where the complexity of
   generating $k$ distinct words using a rejection approach becomes heavily dominated by the rejection step. Finally, the proportion
  of the distribution contained within a sampled set may be affected, positively or negatively, by the adjunction of weights.
  One of the authors proposed a non-redundant version of the recursive method~\cite{Ponty2008b} to work around the first issue.
  However the question of the dependency between the weights and the level of redundancy was left open in a general setting.

%\subsection{Understanding the impact of weights on redundancy}
The aim of the current work is to analyze the redundancy and coverage of a weighted sampled set of words.
To tackle these questions, one can reformulate the repeated generation of words within a weighted language as a
  random allocation of balls into urns. Namely each word $w$ in $\Lang_n$ the restriction of the language to words of length $n$
  will correspond to an urn having probability proportional to the weight of $w$.
  A list of questions naturally arise which can be rephrased into classic random allocations problems:
  \begin{enumerate}
    \item How many words are required before some word is drawn twice?
    This is a weighted instance of the Birthday \emph{paradox} (the first 2-birthday~\cite{FlaGarThi92}).
    \item How many words must be sampled before each word in $\Lang_n$ is encountered at least once?
    One finds in the above formulation the Coupon collector problem.
    \item How many distinct words are there after sampling $k$ words?
    This is equivalent to the expected number of urns having positive load after throwing  $k$ balls.
    \item What is the coverage, i.e. the cumulated weight/probability of a non-redundant sampled set after $k$ generations?
    This last problem rephrases as the cumulated weight/probability of urns having positive load after throwing $k$ balls.
  \end{enumerate}

In this paper, we address and provide closed formulae and/or asymptotic estimates for these four statistical quantities under natural
conditions of non-degeneracy, and illustrate our results with an analysis of a statistical sampling algorithm used to
predict the folding of RNA. After this short introduction we remind in Section~\ref{sec:def} some basic notions related to context-free grammars, languages, algebraic functions
and their weighted analogs. In Section~\ref{sec:results}, we state our main results on weighted context-free languages in the form of four theorems dedicated
to the four questions above. General results on weighted urns models are established or recalled in Section~\ref{sec:theorems}, of which our theorems are direct corollaries.
We apply in Section~\ref{sec:rna} our theorems to an analysis of a statistical sampling algorithm used to predict the functional folding of RNAs,
using the fact that the three-dimensional structure of an RNA can be modeled by a secondary structure, i. e. a word of a Motzkin-like context-free language.
We conclude with some possible extensions of the current work.

\section{Definitions and notations}\label{sec:def}
  \subsection{Weighted context-free languages}
    Throughout the rest of the document, $n$ will stand for the length of generated words.
    For the sake of self-containment, let us start by recalling some definitions found in Denise~\emph{et al}~\cite{Denise2000}.

    A \Def{weighted context-free grammar} $\WGram{\W}$  is a 5-tuple $(\W,\Voc,\NtSet,\ProdRules,\Axiom)$ such that
	\begin{itemize}
		\item $\Voc$ is the alphabet, i.e. a finite set of terminal symbols.
		\item $\NtSet$ is a finite set of non-terminal symbols.
		\item $\ProdRules$ is the finite set of production rules, each of the form $\Nt\Production X$,
		for $\Nt\in\NtSet$ any non-terminal and $X\in \{\Voc\cup\NtSet\}^*$.
		\item $\Axiom$ is the \Def{axiom} of the grammar, i. e. the initial non-terminal.
        \item $\W$ is a positive \Def{weight} vector $\W=(\W_t)_{t\in \Voc}$, assigning positive weights to each letter $\T_i\in\Voc$.
	\end{itemize}
    Let us further assume that the input grammar is unambiguous.
    This is a real limitation, however a similar analysis for intrinsically ambiguous
    languages is rather challenging since the associated generating functions are not necessarily algebraic but possibly transcendental~\cite{Flajolet1987}.

    Let us denote by $\Lang$ be the language generated from the axiom of $\WGram{\W}$, and by $\Lang_n$ its restriction to words of size $n$.
    One can extend the weight multiplicatively on any word $w\in\Lang$ such that $$\WF{w} = \prod_{t \in w} \W_t.$$
    This gives rise to the notion of \Def{weighted generating function} $L_{\W}(z)$ for a context-free language $\Lang$, a natural
    generalization of the ordinary generating function where each word is counted with multiplicity equal to its weight
    $$ L_{\W}(z) = \sum_{w\in\Lang} \WF{w}z^{|w|} = \sum_{n\ge 0} \PF{\W}{n} z^n$$
    where $\PF{n}{\W} = \sum_{w\in\Lang_n} \WF{w}$ is the total weight of words of size $n$.
    In particular, the number $\MCount{n}$ of words of size $n$ can be also expressed as $\MCount{n}=|\Lang_n|=\PF{n}{{\bf 1}}$.

    The weighting scheme then defines a \Def{weighted distribution} on $\Lang_n$ through
    $$ \Prob{w\;|\;n,\W} = \frac{\WF{w}}{\sum_{w'\in \Lang_n} \WF{w'}} = \frac{\WF{w}}{\PF{n}{\W}}.$$

    Finally let us define the \Def{$k$-th moment} of a $\W$-weighted distribution as
  \begin{equation}\Mom_{k,n} = \sum_{i=1}^{\MCount{n}} p_i^k = \frac{\sum_{w\in \Lang_n} \WF{w}^k}{\PF{\W}{n}^k} = \frac{\PF{\W^k}{n}}{\PF{\W}{n}^k}.\label{eq:weightedMoments}\end{equation}

    \subsection{Asymptotics of coefficients}
    The (weighted) generating function of an unambiguous context-free language is a positive solution of an algebraic system of equations,
    therefore its singularities are algebraic.
    Let us first assume that the dominant singularity $\rho_{\W}$ is unique.

    Then, for any fixed $\W$, the coefficients of $L_{\W}(z)$ admit an asymptotic equivalent of the form
      \begin{equation} [z^n]\; L_{\W}(z) = \PF{\W}{n} \sim \kappa_{\W}\cdot\rho_{\W}^{-n}\cdot n^{-k_{\W}}\left(1+\BigO{n^{-k'_{\W}}}\right),\label{eq:genexpansion}\end{equation}
      for $\rho_{\W}\in (0,1]$, $\kappa_{\W}$ some positive real value, and $k_{\W}, k'_{\W}$ some positive rational numbers such that $k'_{\W}>0$.
       The asymptotic equivalent for the number of words $\MCount{n}= |\Lang_n| = [z^n]\; L(z)$ is obtained as a special case of the above, yielding
       \begin{equation} \MCount{n} = |\Lang_n| = [z^n]\; L(z) \sim \kappa\cdot\rho^{-n}\cdot n^{-k}\left(1+\BigO{n^{-k'}}\right)\end{equation}
       with $\rho:= \rho_{{\bf 1}}$, $\kappa:=\kappa_{{\bf 1}}$, $k := k_{{\bf 1}}$ and $k' := k'_{{\bf 1}}>0$ defined as above.

    If the assumption on the unicity of the dominant singularity does not hold, then different singularities may be found on the circle of radius $\rho_{\W}$.
    In this case the coefficients of the generating functions do not admit an universal expansion of the form described in Equation~\ref{eq:genexpansion} since
    the contributions of various singularities may cancel out.
    %However this situation arises from periodicity phenomena within the grammar (see~\cite{Drmota97} for a discussion),
    %and there always exists a modulus $R$ such that, for each $r \in [0,R-1]$, an expansion similar to Equation~\ref{eq:genexpansion} holds for large values of $n$ such that $n \equiv r~[R]$.

    \subsection{Weight classes}
    Let us denote by ${\bf\CW}_n$ the vector of all possible and distinct weights within $\Lang_n$ ordered increasingly ($\CW_{n,i}<\CW_{n,i+1}$).
    In particular, let $\WMin{\W}{n}:=\CW_{n,1}$ (resp. $\WMax{\W}{n}:=\CW_{n,|{\bf\CW}_n|}$) be the \Def{minimal} (resp. \Def{maximal}) \Def{weight} of a word within $\Lang_n$.
    We denote by ${\bf \NCW}_{n,i}\subset \Lang_n$ the class of words having weight ${\CW}_{n,i}$ and by $\MBCount{n}{i}=|{\bf \NCW}_{n,i}|$ its cardinality.
    %Remark that the number $|{\bf \CW}_n|$ of classes of weights  is bounded from above by $(n+1)^{|\Voc|}$.

  \section{Main results}\label{sec:results}
    Let $\WGram{\W}$ be a weighted context-free grammar generating a language $\Lang$, $\W$ its a weight vector and $n\in \mathbb{N}$ a length.
    Remind that $\WMin{\W}{n}$ and $\WMax{\W}{n}$ are the minimal and maximal weight of a word in $\Lang_n$ respectively.
    Let $\rho_{\W}$ be the dominant singularity of $L_{\W}(z)$, and consider the following conditions:
    \begin{enumerate}[label={\bf C\arabic*}]
      \item\label{cond:algebraicSing}Diversity: Let $\PMax{n}{\W}:=\WMax{\W}{n}/\PF{\W}{n}$ be the probability of  the most probable word within $\Lang_n$ with respect
      to a weight function $\W$, then there exists $\beta>1$ such that
      $\PMax{\W}{n} \in \BigO{\beta^{-n}}$.
      \item\label{cond:logPosWeights}Log-positive weights: For each terminal symbol $\T\in\Voc$, $\W_{\T}>1$.
      \item \label{cond:diversity}Bounded dependency: For any rational number $k>1$ and any weight vector $\W$  such that Condition~\ref{cond:logPosWeights} holds, ${\rho_{\W}}^k < \rho_{\W^k}$ holds.
    \end{enumerate}

  \begin{thm}[First collision]\label{th:weightedBirthday}
    Under conditions~\ref{cond:algebraicSing},~\ref{cond:logPosWeights} and~\ref{cond:diversity}, the expected number of generations $E[B_{n,\W}]$ before some word of $\Lang_n$ is drawn twice is such that
      \begin{equation}
      E[B_{n,\W}] \sim \frac{\sqrt{\pi}}{\sqrt{2\Mom_{2,n}}}  = \frac{\PF{\W}{n}\sqrt{\pi}}{\sqrt{2\PF{\W^2}{n}}}  \in \Omega\left(\gamma^n\right),\quad \gamma:=\frac{\rho_{\W}}{\sqrt{\rho_{\W^2}}}>1   \end{equation}
  \end{thm}

  \begin{thm}[Full collection]\label{th:weightedCoupon}
    The expected number of generations $E[C_{n,\W}]$  before all the words in $\Lang_n$ are generated at least once is such that
      \begin{equation}
      \frac{\MPF{\W}{n}}{\WMin{\W}{n}} \le E[C_{n,\W}] \le 2\cdot \Harm{\MCount{n}}\cdot\frac{\MPF{\W}{n}}{\WMin{\W}{n}}
      \end{equation}
    which, for large values of $n$, adopts the equivalent
      \begin{equation}
    \frac{\kappa_{\W}\cdot\rho_{\W}^{-n}}{\WMin{\W}{n}\cdot n^{k_{\W}}} \le E[C_{n,\W}] \le \frac{2\cdot\log(1/\rho)\cdot\kappa_{\W}\cdot\rho_{\W}^{-n}}{\WMin{\W}{n}\cdot n^{k_{\W}-1}}.\\
      \end{equation}

    Moreover in the uniform distribution ($\W ={\bf 1}$) the above expression simplifies into
    \begin{equation}
        E[C_{n,1}]  = \MCount{n}\cdot \Harm{\MCount{n}} \sim \frac{\kappa\cdot\log(1/\rho)\cdot \rho^{-n}}{n^{k-1}}\left(1+\BBigO{1/n^{k'}}\right).
    \end{equation}
  \end{thm}

  \begin{thm}[Distinct samples]\label{th:weightedDistinct}
    The expected number $E[N_{n,\W,k}]$ of distinct words obtained after $k$ generations is such that
      \begin{equation}
      E[N_{n,\W,k}] =  \sum_{i=1}^{|{\bf \CW}|} \MBCount{\W}{i}\cdot\left( 1 - \left( 1 - \frac{\CW_{n,i}}{\PF{\W}{n}} \right)^k \right) =  \sum_{i=1}^{m} \MBCount{\W}{i}\cdot\left(1-e^{-\frac{\CW_{n,i}}{\PF{\W}{n}}k}\right) + \BigO{1}.
      \end{equation}
  \end{thm}

  \begin{thm}[Coverage]\label{th:weightedCoverage}
    In a weighted distribution, the expected cumulated probability $E[P_{n,\W,k}]\in [0,1]$ of the set of distinct words obtained after $k$ generations is given by
      \begin{equation}
      E[P_{n,\W,k}] =  \sum_{i=1}^{|{\bf \CW}|} \MBCount{\W}{i}\cdot\frac{\CW_{n,i}}{\PF{\W}{n}}\cdot\left( 1 - \left( 1 - \frac{\CW_{n,i}}{\PF{\W}{n}} \right)^k \right).
      \end{equation}

    Moreover if Condition~~\ref{cond:algebraicSing} is satisfied, then there exists $\beta>1$ such that, for any $k\in o(\beta^n)$, one has
    \begin{equation} E[P_{n,\W,k}] =  k\cdot\Mom_{2,n} \left( 1 + \BigO{\beta^{-n}}\right).\end{equation}
  \end{thm}
  Remark that there are at most $(n+1)^{|\Voc|}$ different compositions/classes of weights, and therefore Theorems~\ref{th:weightedDistinct} and~\ref{th:weightedCoverage}
  immediately suggest polynomial time algorithms for computing the expected number of distinct words and coverage respectively.

  \subsection{Discussing the loss of generality}
    Let us discuss the loss of generality induced by the above conditions:
    \begin{itemize}
      \item Condition~\ref{cond:algebraicSing} requires that no polynomial group of words contribute asymptotically to a significant
      part of the weighted distribution. This is the typical case in weighted languages, as the exponential growth of $\PF{\W}{n}$ usually arises as a cooperation
      between the natural combinatorial explosion of the numbers of words and their individual weights. However this condition is restrictive,
      and discards languages of polynomial growth, or grammars where a (strongly connected) component of polynomial growth dominates asymptotically.
      \item Condition~\ref{cond:logPosWeights} can be assumed without loss of generality since the weighted distribution is stable through
      the multiplication of all weights by a positive constant.
      \item Condition~\ref{cond:diversity}: Remember that Condition~\ref{cond:algebraicSing}
      implies that there exist some constants $C>0$ and $\beta>1$ such that $\W(w) \le C\cdot\PF{\W}{n}/\beta^{n}$ for all $w\in \Lang_n$.
      It follows that, for all $k>1$,
      \begin{align*}
      \PF{\W^k}{n} &= \sum_{w\in \Lang_n}\W(w)^k \le \sum_{w\in \Lang_n}\W(w)\cdot \left(\frac{C\cdot\PF{\W}{n}}{\beta^{n}}\right)^{k-1} = \PF{\W}{n}^k\cdot \frac{C^{k-1}}{\beta^{(k-1)\cdot n}}.
       \end{align*}
       Consequently the exponential growth factor $\rho_{\W^k}^{-1}$ of $\PF{\W^k}{n}$ is such that
       $$\rho_{\W^k}^{-1} \le \left(\beta^{(k-1)}\rho_{\W}^k\right)^{-1}<{\rho_{\W}^{k}}^{-1}$$
       and Condition~\ref{cond:diversity} is a direct consequence of Condition~\ref{cond:algebraicSing}.
    \end{itemize}

\section{General theorems}\label{sec:theorems}
  In the following section we establish general results on non-uniform urn models, which we apply to weighted distributions.
  Let ${\bf u}$ be a set of urns, $m = |{\bf u}|$ its cardinality and, for each $u_i\in {\bf u}$, let $\CW_{i}$ be the weight of $u_i$,
  and $p_{i}$ its probability. This defines a probability distribution ${\bf p} = (p_i)_{i=1}^m$ such that $\sum_{i = 1}^{m} {p_{i}} = 1$ and for all $i\in[1,m-1]$, $p_{i} \le p_{i+1}$.

\subsection{Birthday paradox: First collision}

\begin{thm}\label{th:birthday}
Assume there exists $\tau:=\tau({\bf p})$ such that
  \begin{enumerate}[label=({\Alph*})  ]
    \item $p_{m} \cdot\tau < 1$;\label{cond:smallLargestProba}
    \item $\sqrt{\Mom_{2}} \cdot\tau \rightarrow + \infty$ when $m\to \infty$;\label{cond:secondMoment}
    \item $\sqrt[3]{\Mom_{3}} \cdot\tau \rightarrow 0$ when $m\to \infty$;\label{cond:thirdMoment}
  \end{enumerate}

Then the waiting time $E(B)$ of the first birthday can be approximated by
\[
E(B) = \sqrt{\frac{\pi}{2\alpha_2}} (1+o(1)).
\]
\end{thm}

  \subsubsection{Application to weighted distribution}

  \begin{prop}\label{prop:birthday}
    Let $\WGram{\W}$ be a weighted context-free grammar and $\Lang$ be its associated language, satisfying Conditions~\ref{cond:algebraicSing},~\ref{cond:logPosWeights}, and~\ref{cond:diversity}.
    Then the weighted distribution induced on $\Lang_n$ satisfies the conditions~\ref{cond:smallLargestProba},~\ref{cond:secondMoment} and~\ref{cond:thirdMoment} of Theorem~\ref{th:birthday} for any $\tau_n := \Mom_{k,n}$ such that $2<k<3$.
    Consequently the first collision is observed after $E[B\;|\;n]= \sqrt{\pi/2\Mom_{2,n}} (1+o(1))$ generations.
  \end{prop}

\subsection{Coupon collector: Waiting for the full collection}
  First let us remind that the uniform case is covered by the following \emph{folklore} theorem~\cite{FlaGarThi92}.
  \begin{thm}
    In the uniform distribution, the waiting time $E[C_{{\bf 1}}]$
    is given by
    \begin{equation}
       E[C_{{\bf 1}}] = m\cdot\mathcal{H}_{m} \in \Theta(m\cdot\log(m)).
    \end{equation}
  \end{thm}

  \begin{thm}
    In a non-uniform distribution and for large values of $n$, the waiting time $E[C_{\W}]$ of the full collection obeys
    \begin{equation}
      \frac{1}{p_1} \le E[C_{\W}] \le 2\cdot \Harm{m}\cdot\frac{1}{p_1}\label{eq:couponCollector}
    \end{equation}
    where $p_1$ is the smallest probability of an urn.
  \end{thm}
  \begin{proof}

  First let us point out that, for any urn $u$, the waiting time of the full collection is greater than the expected
  time when a first ball reaches $u$. Since the least probable urn has probability $p_1$,
  then the lower bound on $E[C_{\W}]$ immediately follows.

  From a recent contribution by Berenbrink and Sauerwald~\cite{Berenbrink2009}, we know that the waiting time $E[C_{\W}]$
  for the full collection of $m$ items drawn with respective probabilities $p_1\le p_2\le\ldots\le p_m$
  can be approximated within a $\BigO{\log \log m}$ factor by an estimate
  \begin{equation} \Est_m = \sum_{i=1}^m \frac{1}{i p_i}.\label{eq:estimateCouponCollector}\end{equation}
  More precisely it is shown in~\cite{Berenbrink2009} that
  \begin{equation} \frac{\Est_m}{3e\cdot\log \log m} \le E[C_{\W}] \le 2 \Est_m.\end{equation}

  In our urn model, equation~\ref{eq:estimateCouponCollector} specializes into
  \begin{align*}
    \Est_{m} & = \sum_{i=1}^m \frac{1}{i p_i} = \frac{1}{p_1}\left(\sum_{i=1}^m \frac{1}{i \Delta_i}\right)
  \end{align*}
  where $\D_{i}:=p_i/p_1$. Since $p_1$ is the weight of the least probable urn, then one has $\D_{i}\ge 1,\; \forall i\in[1,m]$, and therefore the following upper bound holds
  \begin{align*}
    \Est_{m} \le \frac{1}{p_1}\left(\sum_{i=1}^m \frac{1}{i}\right) = \frac{1}{p_1}\cdot\Harm{m}
  \end{align*}
  in which one recognizes the upper bound of Equation~\ref{eq:couponCollector}.
  \end{proof}
  Experiments suggest that the upper bound is in fact reached. For instance, Figure~\ref{fig:CouponBerenbrink} shows the value
  $p_1\cdot\Est_{m}$ for weighted Motzkin paths, where a weight $W>1$ is associated to horizontal steps, while up and down steps remain unweighted.
  In such a case the growth of $p_1\cdot\Est_{m}$ appears to be linear with different slopes depending on the parity of $n$.
  This phenomenon is due to the fact that the minimal number of horizontal steps in a Motzkin word of length $n$ is $0$ (resp. $1$) for even (resp. odd) lengths,
  leading to minimal weights of $1$ for even lengths and $\W$ to odd ones.

  \begin{figure}
  \includegraphics[width=25em]{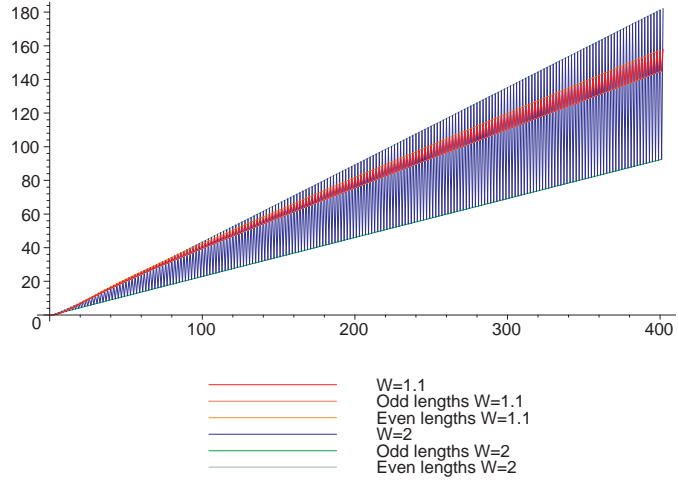}
  \caption{Plots of $p_1\cdot\Est_{m}$ for weighted Motzkin words exhibit a linear growth  on $n$, suggesting that the upper bound is reached.}\label{fig:CouponBerenbrink}
  \end{figure}
\subsection{Occupancy analysis}

   Figuring out the average state after $k$ generations turns out to be easier that the inverse problem -- finding expected number $k$ of generations
   before a given state is observed. We refer to a survey~\cite{Gardy2002} by one of the authors for examples of urns model in the context of the analysis of algorithms.
   Here we establish a general formula for the cumulated weight in a weighted urn model through a generating function analysis.
  \begin{thm}
    The total weight $E[W_{k}]$ of occupied urns after throwing $k$ balls is given by
    \begin{equation}E[W_{k}] = \sum_{i=1}^{m} \CW_{i}\cdot\left(1-\left(1-p_{i}\right)^k\right).\label{eq:occupancy}\end{equation}
  \end{thm}
  \begin{proof}
    Consider the bivariate generating function $$\Psi_{\W}(x,y) =  \sum_{j\ge 0} \sum_{k\ge 0} a_{j,k}\cdot x^j\cdot \frac{y^k}{k!}$$
    where $a_{j,k}$ is the probability of reaching a set of urns having cumulated weight equal to $j$ upon throwing $k$ balls.
    Remark that such random allocations can be reinterpreted as sequences of $m$ urns, each urn $u_i$ containing either at non-empty
    set of balls (associated with a $x^{\CW_{i}}(e^{p_{i}y}-1)$ contribution) or no ball ($y^0=1$). Consequently the generating function $\Psi_{\W}(x,y)$ can be reformulated as
    $$ \Psi_{\W}(x,y) = \prod_{i=1}^{m}\left(1+x^{\CW_{i}}\left( e^{p_{i} y} -1 \right)\right).$$
    The generating function for the expectation of weight is then classically obtained through a partial derivative on $x$.
         \begin{align*}
   E[W_{k}] &
   = \left[\frac{y^k}{k!}\right]\frac{\partial \Psi_\W(x,y)}{\partial z}(1,y) = \left[\frac{y^k}{k!}\right] e^{-y}\sum_{i=1}^{m} \CW_{i}\cdot\left(1-e^{-y\,p_{i}}\right)\\
   & = \sum_{i=1}^{m}\cdot\left(\left[\frac{y^k}{k!}\right] e^y - \left[\frac{y^k}{k!}\right]e^{y\,(1-p_{i})}\right) = \sum_{i=1}^{m} \CW_{i}\cdot\left(1 - (1-p_{i})^k\right)
   \end{align*}
  \end{proof}
   Remark that, upon setting $\CW_{i} = 1$, Equation~\ref{eq:occupancy} simplifies into $E[N_{k}]$ of urns reached by at least one ball (cf Hwang and Janson~\cite{Hwang2008}), such that
    \begin{equation*}E[N_{k}] = \sum_{i=1}^{m} \left(1-\left(1-p_{i}\right)^k\right) = \sum_{i=1}^{m} \left(1-e^{-p_{i}k}\right) + \BigO{1}\end{equation*}

\subsubsection{Asymptotic estimates for the coverage}
Let us start from the formula
\[
E[W_{k}] = \sum_{i=1}^{m} \CW_{i}\cdot\left(1 - (1-p_{i})^k\right) = \sum_{i=1}^{m} \CW_{i}\cdot\left(1 - e^{k\cdot\log(1-p_i)}\right).
\]
Since $p_i <1$ for all $i\in[1,m]$, then one can use an approximation $\log(1-p_i) = - p_i + \BigO{p_i^2}$
for large values of $m$, which can be be injected into $E$ to obtain
\[
E[W_{k}] =  \sum_{i=1}^{m} \CW_{i}\cdot\left(1 - e^{k\left(- p_i + \BigO{p_i^2}\right)}\right).
\]

%Let us discuss the nature of the dominant term depending on the value of $k$.
If $k\cdot p_m \in o(1)$, then $k\cdot p_i \le k\cdot p_m \in o(1)$ for all $i\in[1,m]$,
and therefore $e^{k (-p_i +  \BigO{p_i^2}) } = 1 - k p_i  + \BigO{kp_i^2},$ which  gives
\begin{equation}
E[W_{k}] = \sum_{i=1}^{m} \CW_{i}\left(  k p_i  + \BigO{kp_i^2} \right)
= k \sum_{i=1}^{m} \CW_i p_i + \BBigO{k \sum_{i=1}^{m}\CW_i\cdot p_i^2}.\label{eq:coverage}
\end{equation}
%As $k$ increases, the asymptotic equivalent for the exponential will become invalid since $k\cdot p_m$ will at some point be larger than~1.
In weighted languages that satisfy Condition~\ref{cond:algebraicSing}, there exists $\beta>1$ such that $p_i \in \BigO{\beta^{-n}}$, for all $i\in[1,m]$.
Consequently, for any $k\in o(\beta^n)$, one has
$$E[W_{k}] = k \sum_{i=1}^{m} \CW_i p_i \left( 1 + \BigO{\beta^{-n}}\right) = k\cdot\PF{\W}{n}\cdot \Mom_{2,n} \left( 1 + \BigO{\beta^{-n}}\right).$$

\begin{comment}
%Yann: Les probabilités classées par ordre décroissant ne peuvent pas satisfaire cette propriété dans les langages
For instance let us assume that there exists exactly $M$ urns such
that $p_{m-M}=p_{m-M+1}=\ldots=p_m$. Consider $k$ such that $k p_i \rightarrow + \infty$ for all $i\ge n-M$
and $k p_i \rightarrow 0$ for $i < n-M$. In this domain,
\[
E[N_{k}] = k \mu \alpha_2 - k  \frac{m_1 W_1^2}{\mu} + m_1 W_1 \left( 1 - e^{- k W_1/\mu (1+o(1))} \right) + O (k \alpha_3/\mu) .
\]
Assume that we also have $m_1 W_1 e^{-kW_1/\mu} \rightarrow 0$ (this depends on $m_1$ and $W_1$); then
\[
E[N_{k}] \sim k \mu \alpha_2 - k  \frac{m_1 W_1^2}{\mu} + m_1 W_1 +o(1) .
\]
Taking larger $k$ will introduce more weights $W_i$ into the formula for $E$.

When $k$ becomes large enough, i.e. for $k > \mu/W_n$, then all the terms $e^{- k W_i/\mu}$ are small and we expect to get back $E[N_{k}] = \mu$.
(est-ce que cela a un sens???)
\end{comment}

%\subsubsection{Coverage of weighted distributions}
%Assuming the $p_i$ are induced by a weighted distribution such that $p_i = \CW_{n,i}/\PF{n}{i}$, then Equation~\ref{eq:coverage} can
%be expressed in term of the moments
%\[
%E[N_{k}] = k\cdot\PF{n}{i}\cdot\Mom_{2,n} + \BigO{k\cdot\PF{n}{i}\cdot\Mom_{3,n}} = k\cdot\PF{n}{i}\cdot\Mom_{2,n} (1+o(1)).
%\]
%

\section{Application to the statistical sampling of RNA}\label{sec:rna}

\pgfdeclarelayer{background}
\pgfdeclarelayer{foreground}
\pgfsetlayers{background,main,foreground}
\begin{figure}
\tikzstyle{st}=[rectangle,draw=gray!80,fill=white]
\usetikzlibrary{shadows}
\begin{tikzpicture}[scale=1,node distance=0mm, inner sep=0pt]
  \node  at (0,0) {\includegraphics[width=130pt]{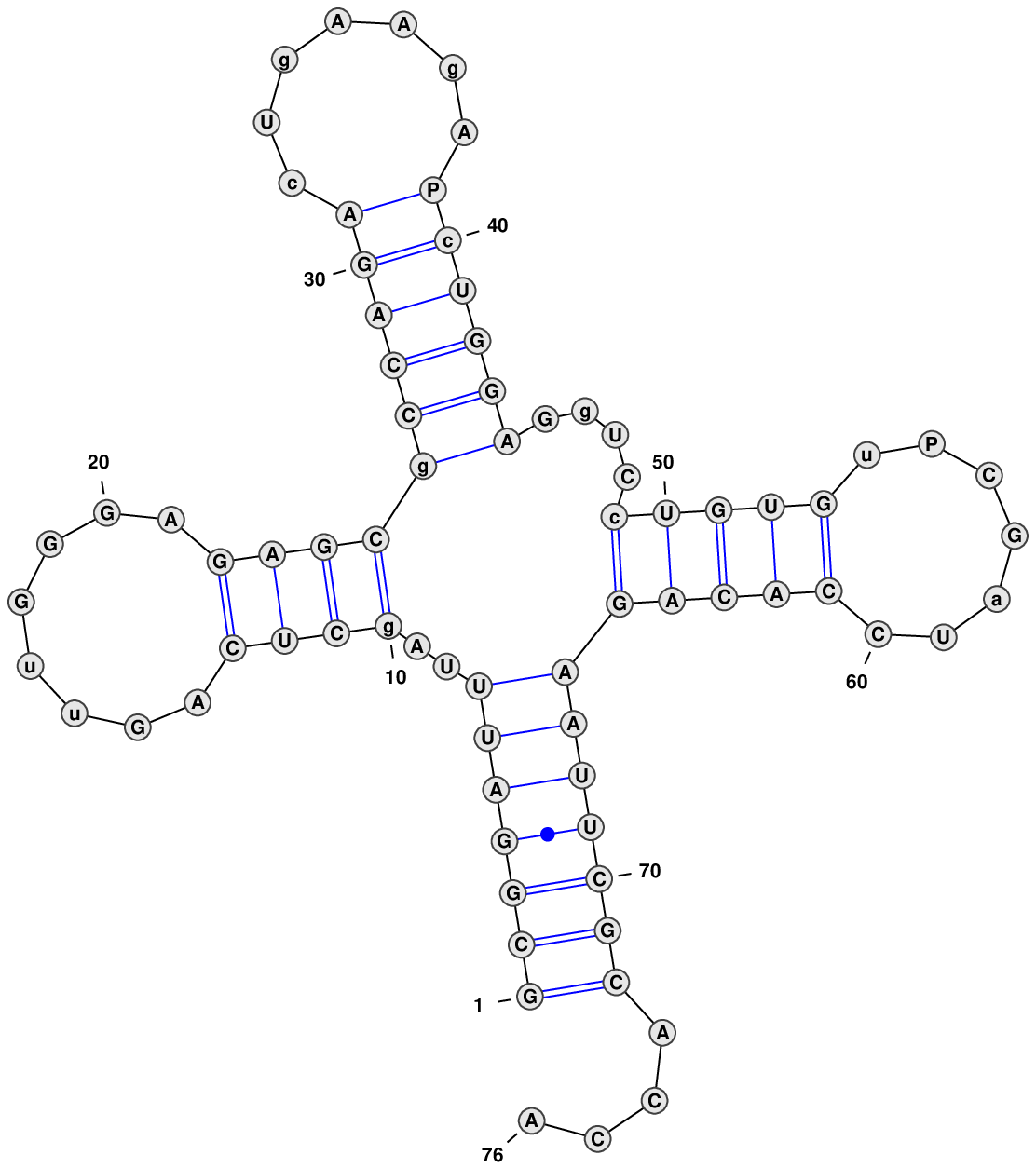}};
  \node  at (7.0,-1.7) {\includegraphics[width=220pt]{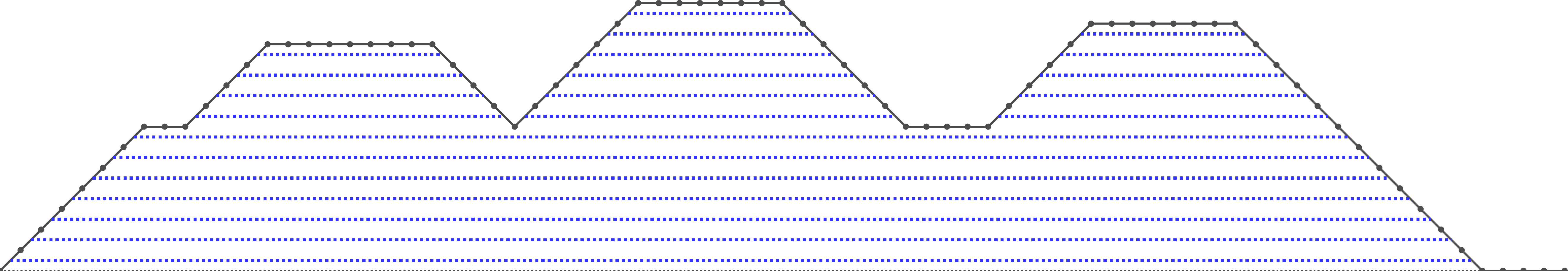}};
  \node[st]  at (4.2,1.5) (s1)     {\includegraphics[width=25pt]{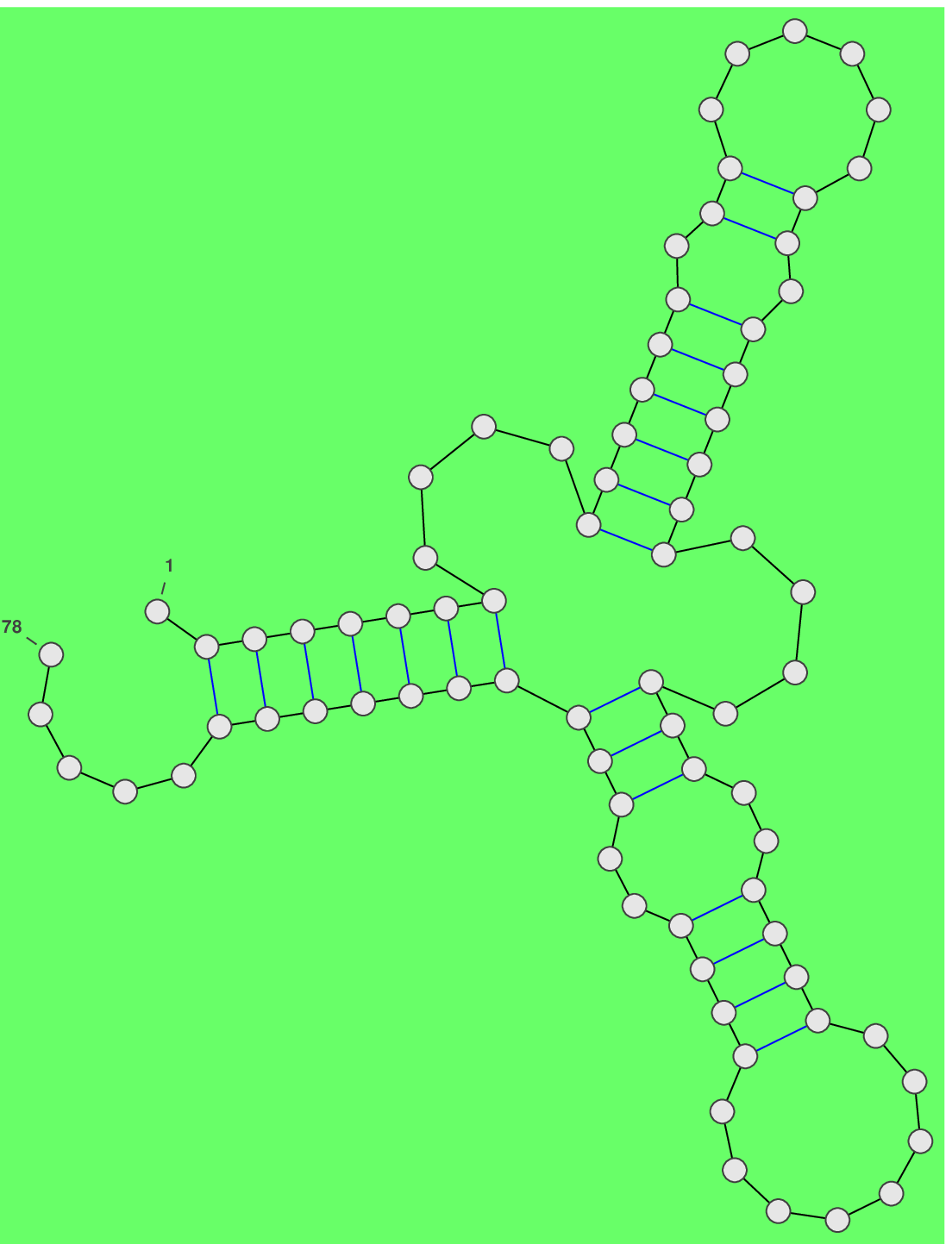}};
  \node[st,right=5pt of s1]  (s2)  {\includegraphics[width=25pt]{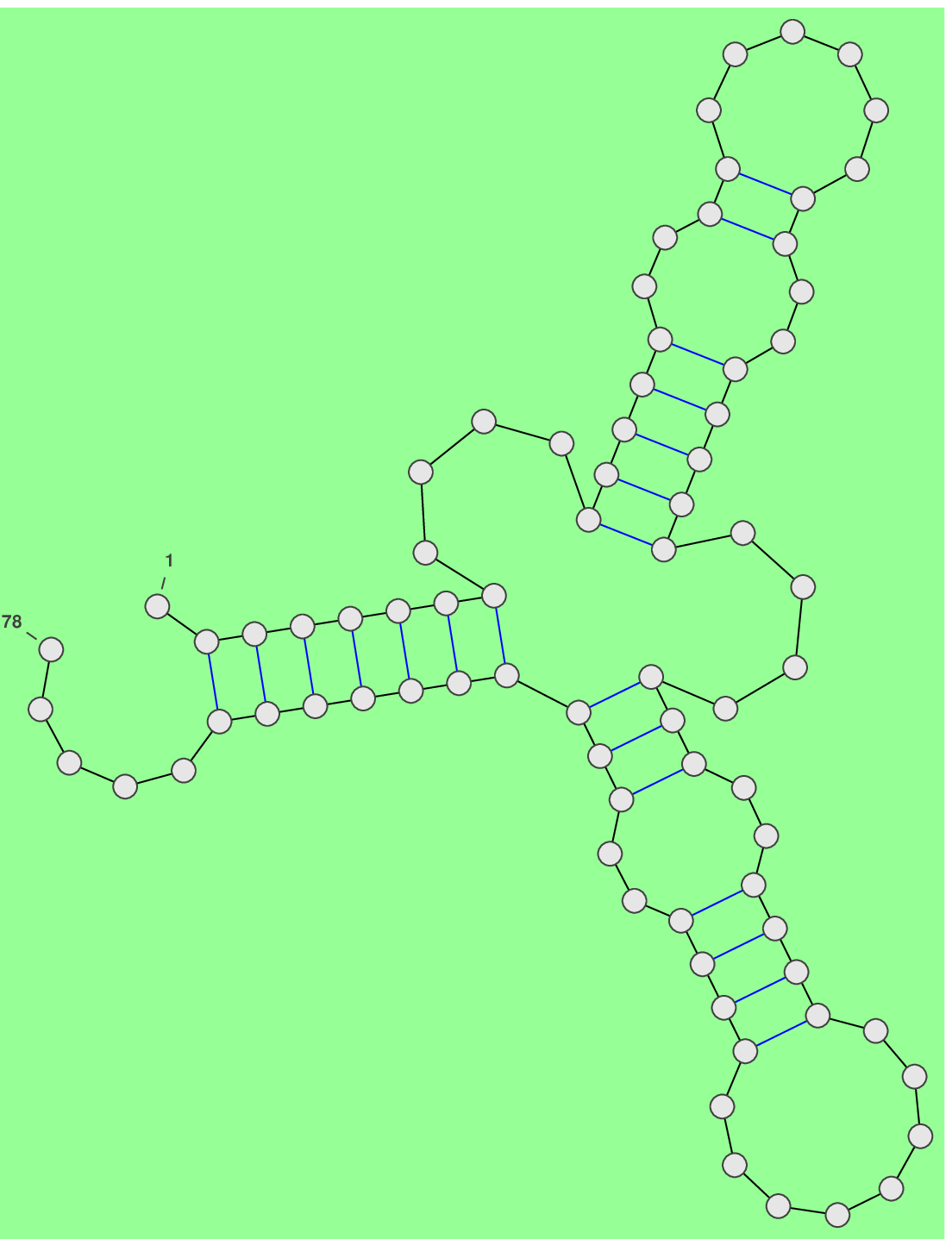}};
  \node[st,right=5pt of s2]  (s3)  {\includegraphics[width=25pt]{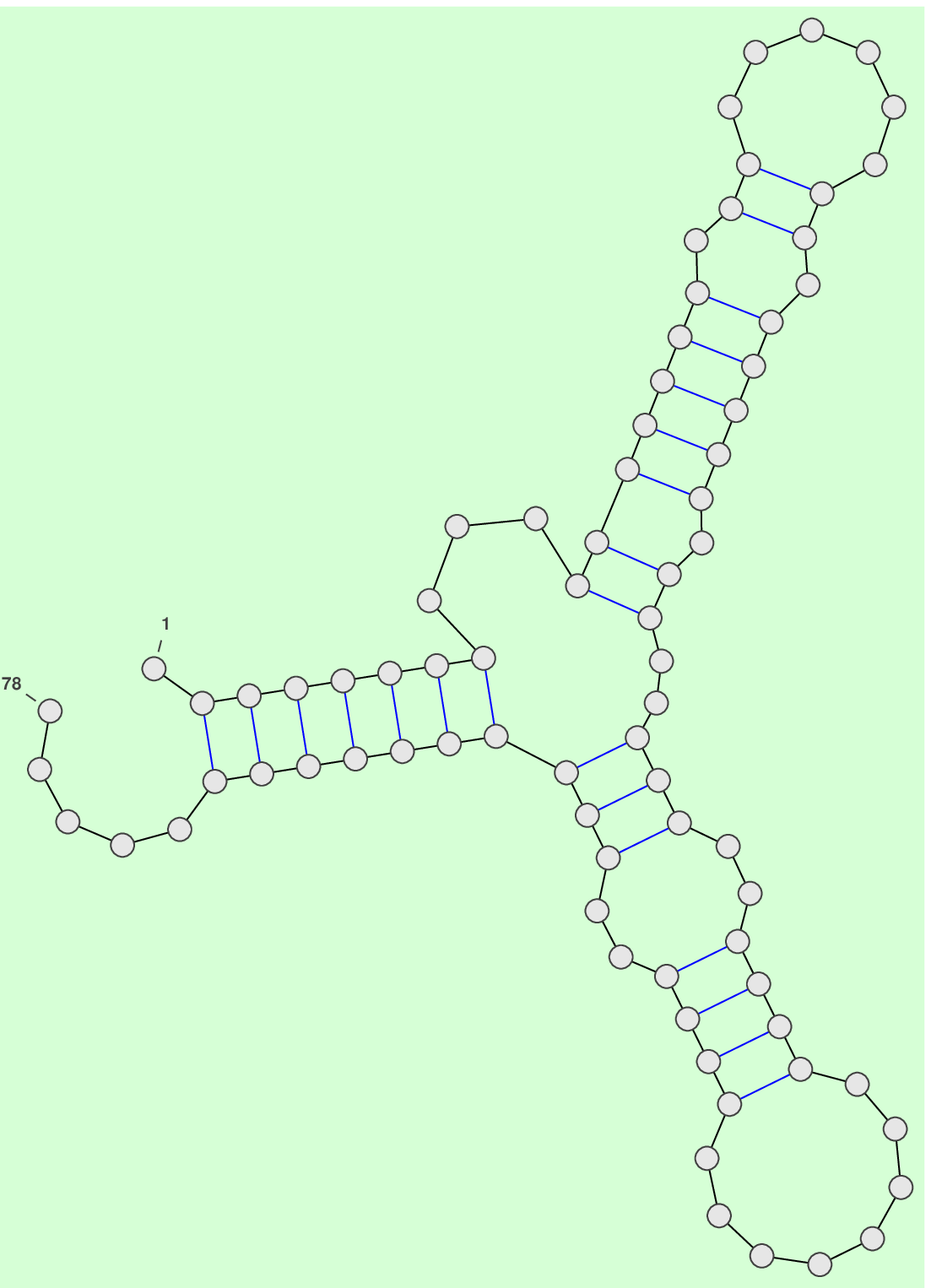}};
  \node[st,right=5pt of s3]  (s4)  {\includegraphics[width=25pt]{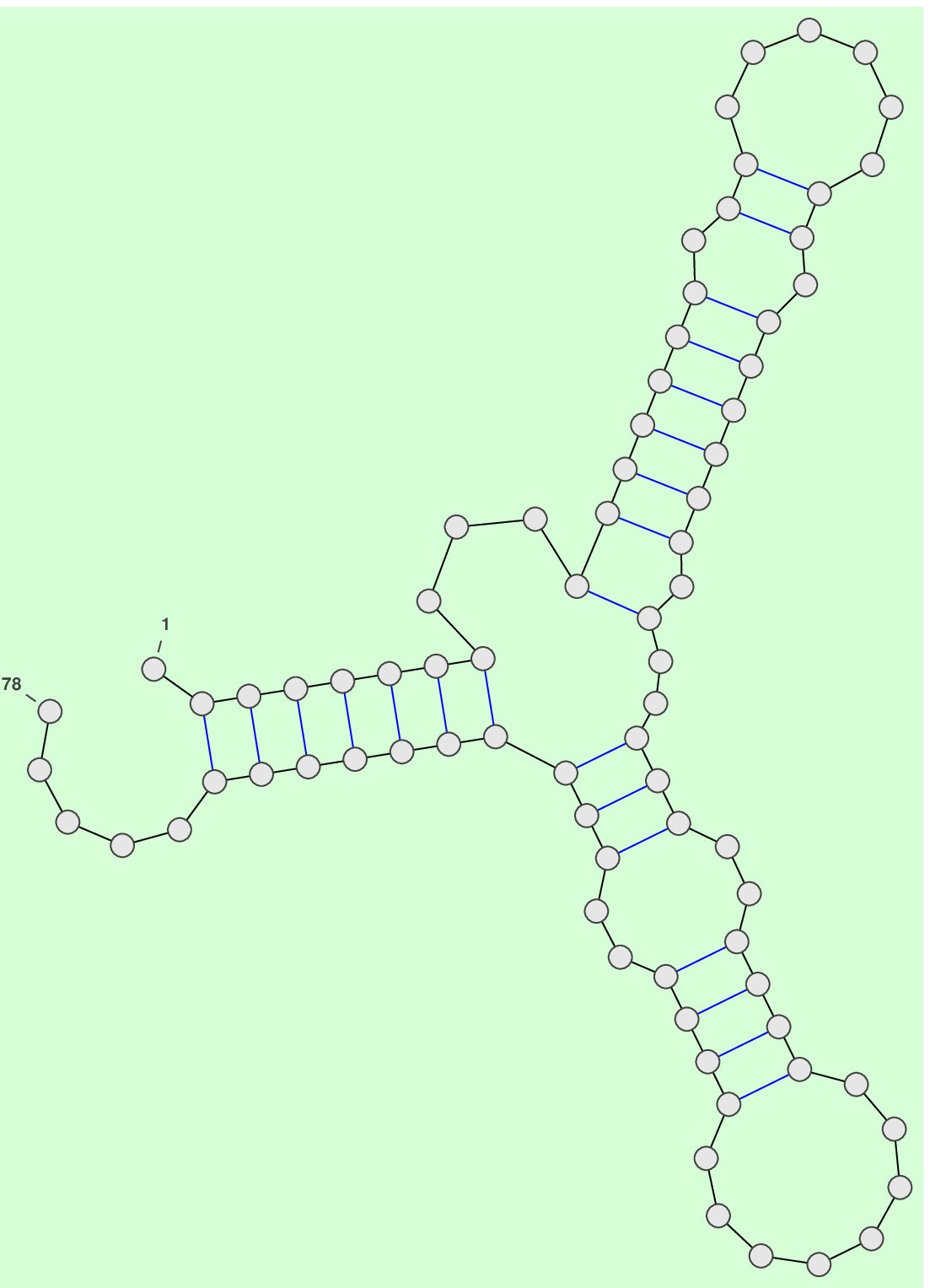}};
  \node[st,right=5pt of s4]  (s5)  {\includegraphics[width=25pt]{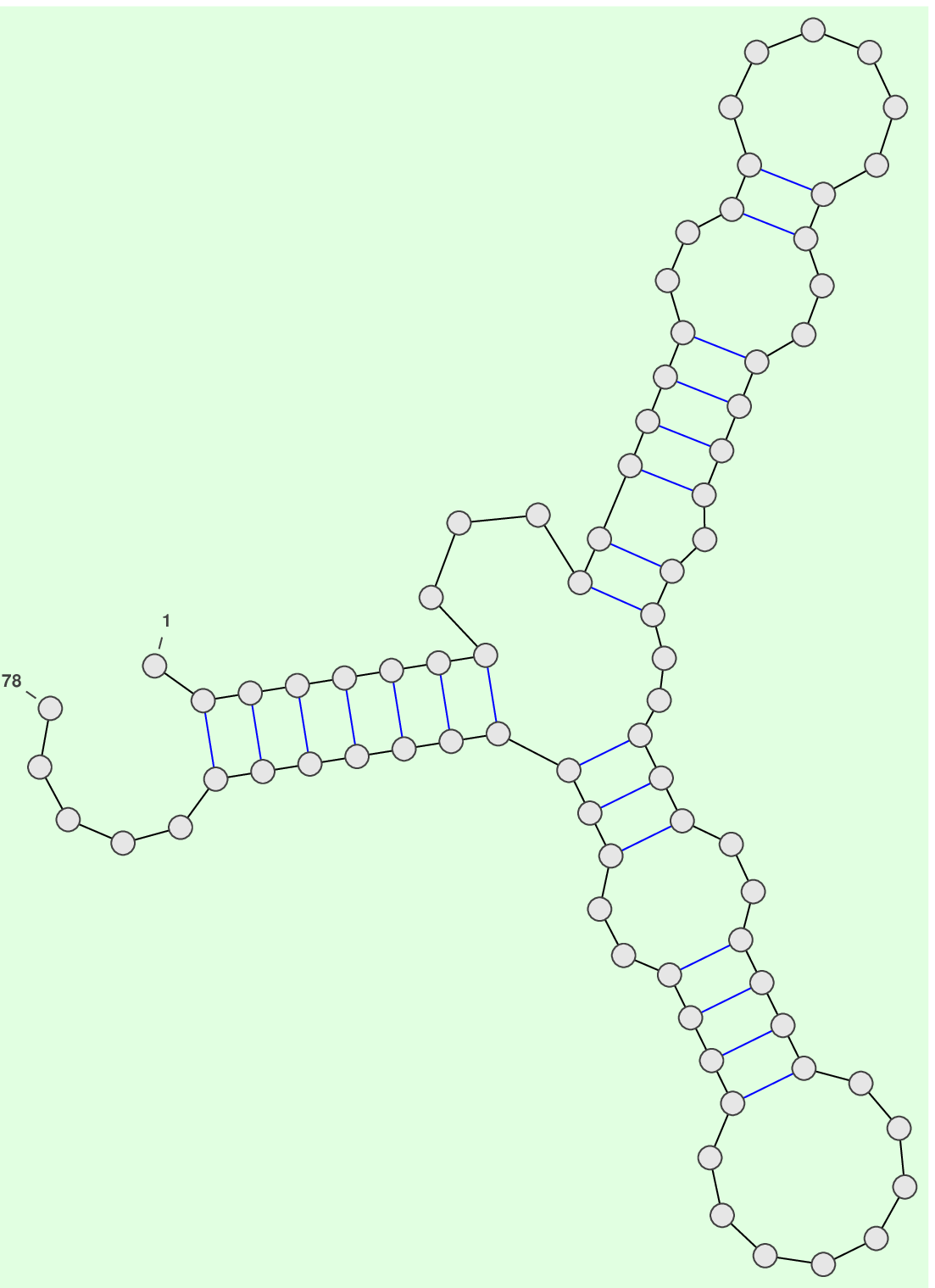}};
  \node[st,right=5pt of s5]  (s6)  {\includegraphics[width=25pt]{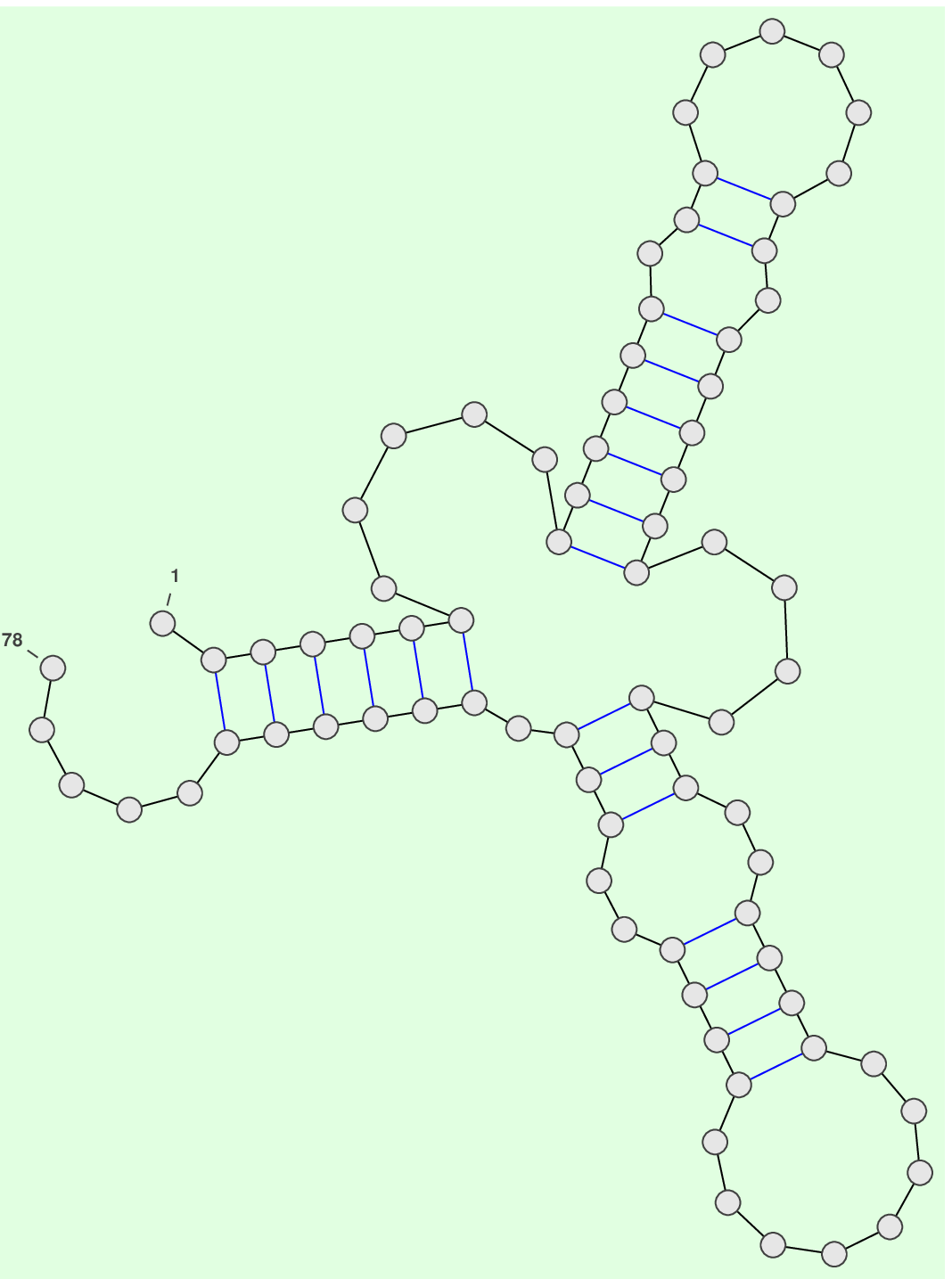}};
  \node[st,below=5pt of s1]  (s7)  {\includegraphics[width=25pt]{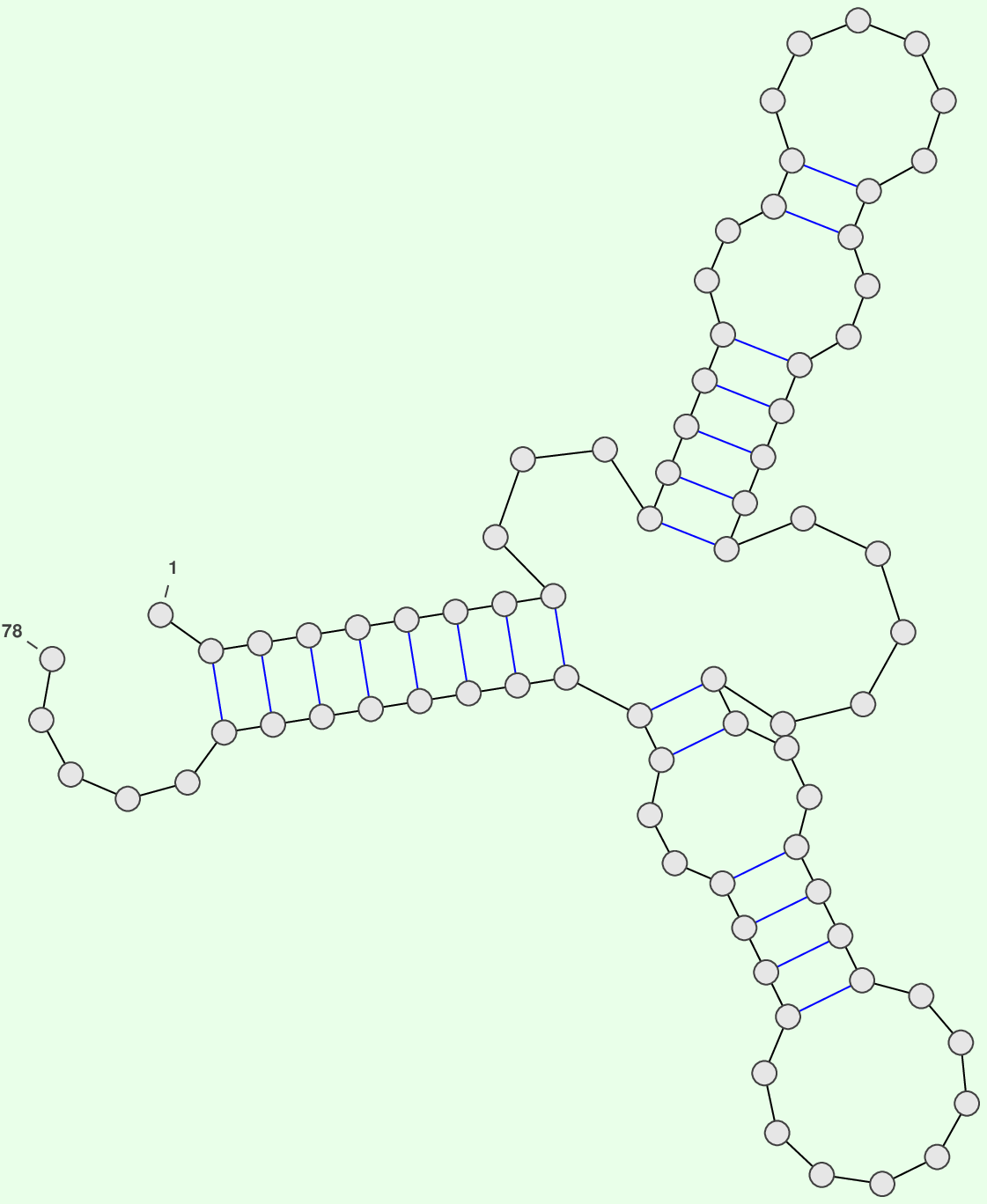}};
  \node[st,right=5pt of s7]  (s8)  {\includegraphics[width=25pt]{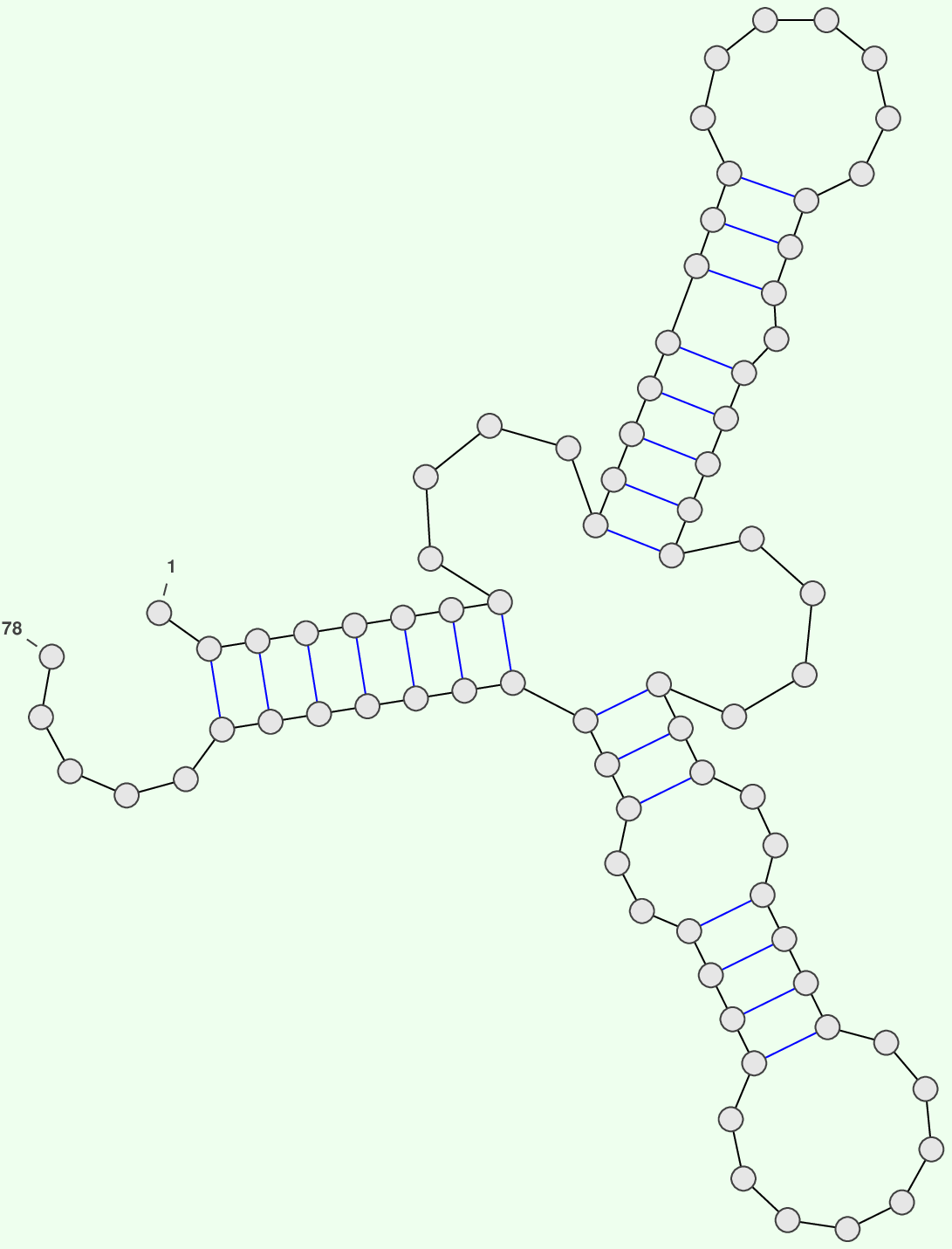}};
  \node[st,right=5pt of s8]  (s9)  {\includegraphics[width=25pt]{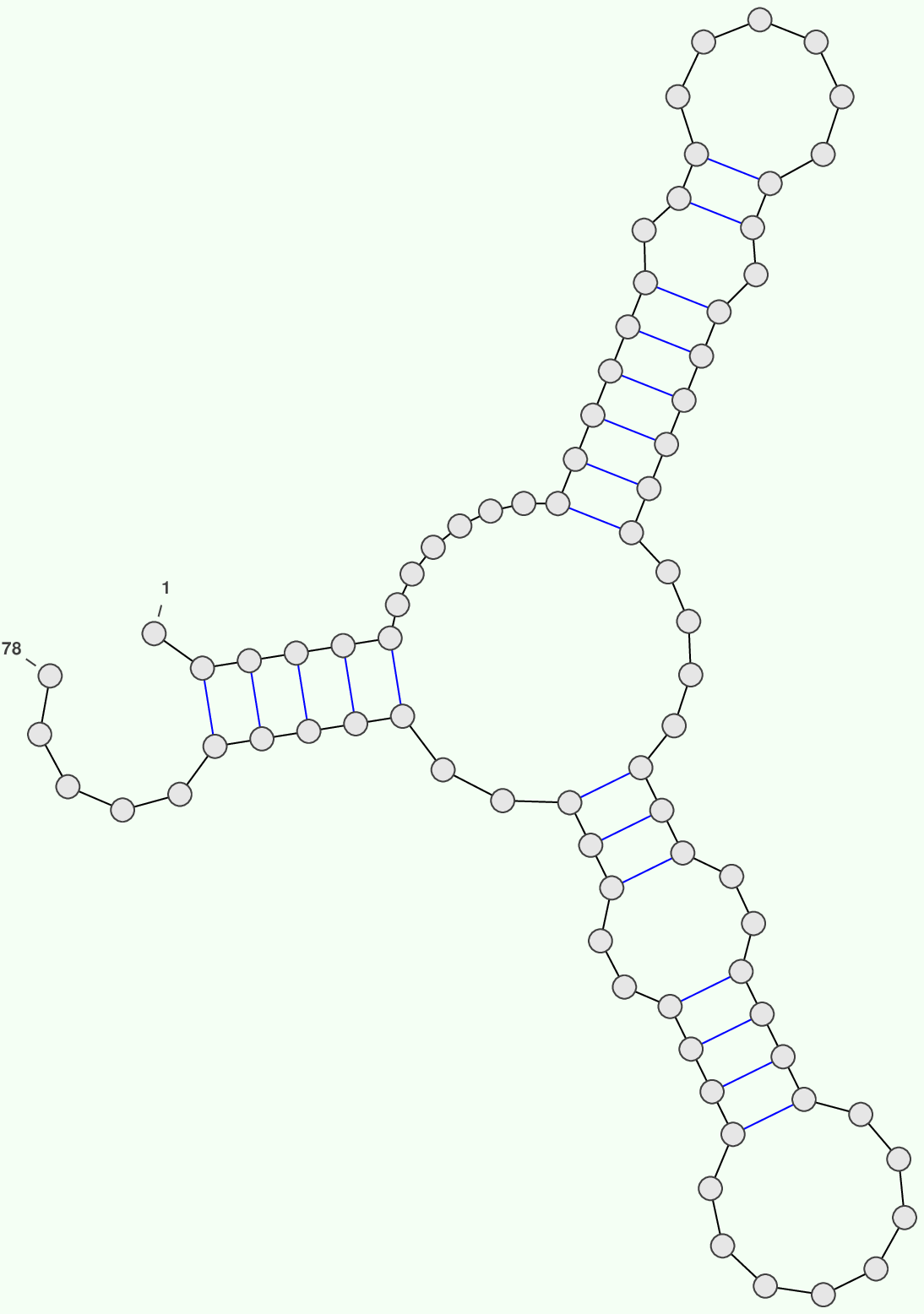}};
  \node[st,right=5pt of s9]  (s10) {\includegraphics[width=25pt]{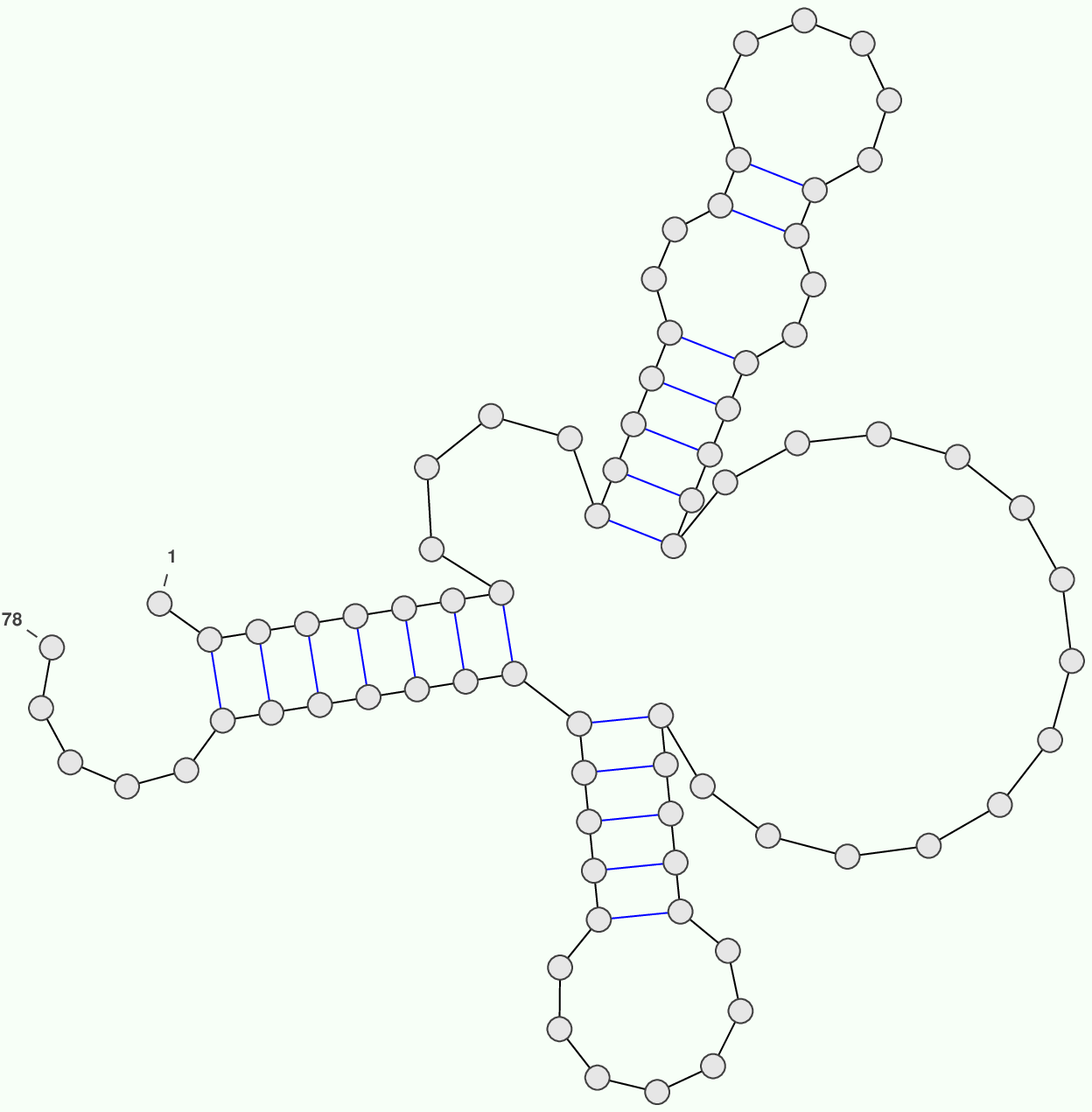}};
  \node[st,right=5pt of s10] (s11) {\includegraphics[width=25pt]{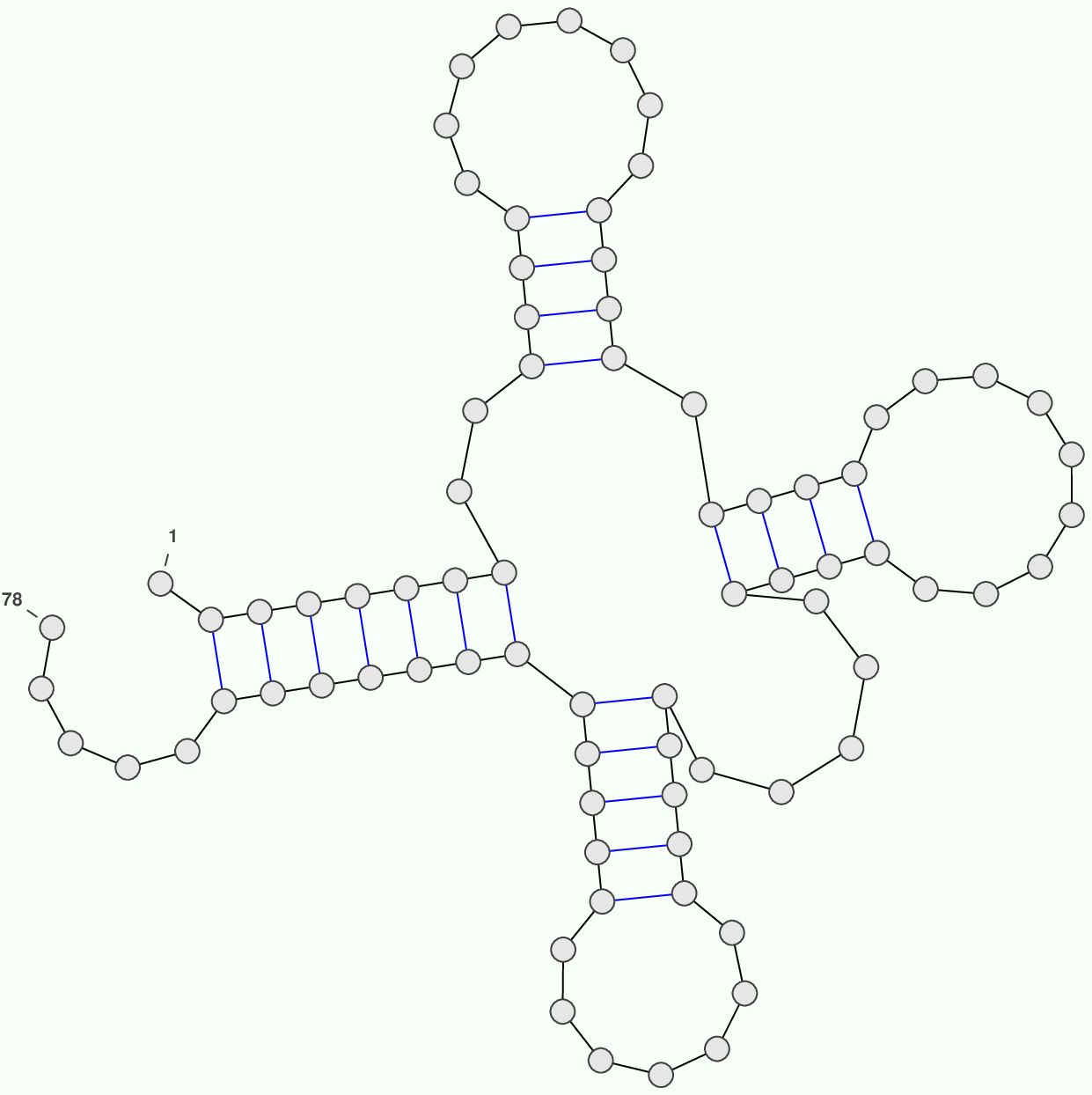}};
  \node[st,right=5pt of s11] (s12) {\includegraphics[width=25pt]{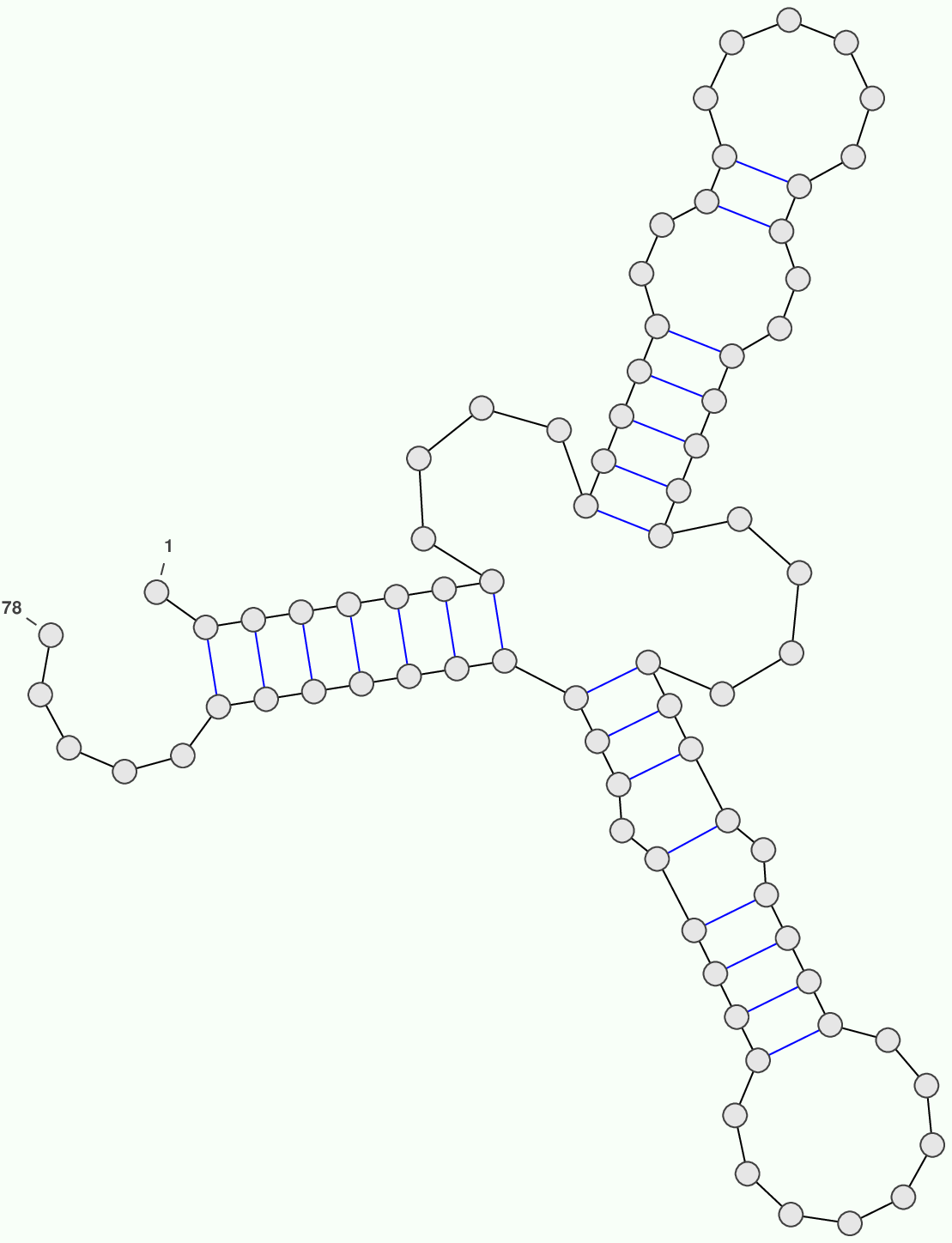}};
\begin{pgfonlayer}{background}
  \node[fill=white,inner sep =5pt, draw=gray, fit=(s1) (s6) (s12)] {};
\end{pgfonlayer}
  \end{tikzpicture}
  \caption{Secondary structure (Left) of a transfer RNA (tRNA) and its equivalent representation as a Motzkin walk (Bottom-right). Top-right: Typical
  picture of the Boltzmann ensemble, i.e. a set secondary structures compatible with the RNA sequence, colored
  according to their respective Boltzmann factor $e^{\frac{E_s}{RT}}$.\label{fig:rna}}
\end{figure}

\subsection{Motivation}
Random generation has recently found a novel application in the \emph{in silico} prediction of RNA folding.
Namely a state-of-the-art method~\cite{DinChanLaw05} for predicting the functional folding of a given RNA sequence uses a
non-uniform random generation scheme~\cite{DiLa03}. This method aims at predicting the functional, or \Def{native},
secondary structure of an RNA, a coarse-grain representation of the three-dimensional conformation.
Based on the observation that the native structure is not necessarily that of lowest free-energy, Ding~\emph{et al} used
a model initially proposed by Mc Caskill~\cite{McCaskill1990}, and hypothesized a Boltzmann distribution based on the free-energy over the set of possible conformations.
Their method generates a representative set of 1000 secondary structures using a statistical sampling algorithm~\cite{DiLa03}.
These structures are then clustered and a consensus structure is extracted. Considering this consensus led to a
better sensibility/specificity tradeoff than previous approaches
based on free-energy minimization~\cite{Zuker1981}.

However, given the variability in length and sequence composition of real RNAs, the 1000 structures criterion seems somewhat arbitrary
and may lead to irreproducible observations in the context of highly variable observables. On the other hand, the sampled sets
of structures might feature a large level of redundancy. Our theorems provide useful tools for a quantitative characterization of such situations.

\subsection{Statistical sampling of RNA secondary structures}
An RNA sequence can be encoded by a sequence of bases A, C, G and U where local compatibility rules (A$\leftrightarrow$U, A$\leftrightarrow$U, and G$\leftrightarrow$U)
allow for a folding, i.e. a formation of chemical bounds between pairs of bases.
The RNA secondary structure constitutes a restriction of all possible base-pairings, where each base is involved in at most one base-pairs with the additional constraint
that the induced matching does not feature crossing interactions.
A simplified energy model of Nussinov~\cite{nussinov80} assigns free-energies contributions $E_{b}$
between $-3.0$ and $-1.0$ KCal.Mol$^{-1}$ to each base-pairs $b$, depending on the number of hydrogen bonds involved in the interaction.
The total free-energy $E_s = \sum_{b\in s} E_{b}$ of a secondary structure $s$ is then inherited additively,
and each secondary structure $s$ is drawn with probability proportional to its \Def{Boltzmann factor} $e^{\frac{E_s}{RT}}$
where $R$ is the perfect gaz constant and $T$ the temperature in Kelvin.

\subsection{Statistical sampling as a weighted generation}
Let us first remind that Motzkin words are well-parenthesized words featuring any number of dots characters $\UB$.
Let us define a \Def{peak} as an occurrence of a motif $\OP\;\CP$, and a $k$-\Def{plateau} as an occurrence of a motif $\OP\UB^k\CP$, $k>0$.
Let $\theta\in\mathbb{N}$ be a parameter, then one defines secondary structures as \emph{peakless} Motzkin words, or
more generally as Motzkin words that are free of $t$-plateaux, for any $t<\theta$. The correspondence between coarse-grained
conformations and Motzkin words is illustrated in Figure~\ref{fig:rna}.
 Each pair of matching parentheses represents a base-pair, and the $\theta$ constant models steric constraints and is typically set to $1$ in combinatorial
 studies~\cite{waterman78} and to $3$ in most RNA folding software. Through an adaptation of Viennot~\emph{et al}~\cite{Chaumont1983}, secondary structures
 can be generated from a non-terminal $S$ using rules
\begin{align*}
S &\to \OP S_{\ge \theta}\CP S\;|\;\UB S\;|\;\varepsilon & S_{\ge \theta} & \to \OP  S_{\ge \theta} \CP S\;|\;\UB S_{\ge \theta} \;|\;\UB^\theta.
\end{align*}

\subsection{Expected times for first collision and full collection}
Assuming a standard homopolymer model, in which any pair of base can bind, statistical sampling is equivalent to a weighted random generation,
taking $\RNAW:=e^{\frac{E}{RT}}$ as the weight of any base-pair $b$ (e.g. any occurrence of a opening parenthesis). The resulting weighted generating function is then given by
\begin{align*}
  \RNASG{\RNAW}{\theta}(z) =&  \frac{1-2z+(w+1)z^2-wz^{\theta+2}-\sqrt{\Delta_{\RNAW, \theta}}}{(1-z)2z^2}\\
   \Delta_{\RNAW, \theta} :=& 1-4z+(6-2w)z^2+4(w-1)z^3+(w-1)^2 z^4\\
   &-2wz^{\theta+2}+4wz^{\theta+3}-2w(1+w)z^{\theta+4}+w^2z^{2\theta+4}.
\end{align*}

Using our formulae, one can get estimates for the waiting times $E[B_{n,\theta,E}]$ and $E[C_{n,\theta,E}]$ for the first collision
and full collection respectively, and observes the following behaviors
\begin{align*}
  E[B_{n,1,-1}] &\sim \frac{1.24\cdot1.54^n}{\sqrt[4]{n^3}} &
  E[B_{n,3,-3}]&\sim \frac{0.85\cdot1.105^n}{\sqrt[4]{n^3}} \\
  \frac{0.64\cdot 4.33^n}{n\sqrt{n}} &\lesssim  E[C_{n,1,-1}] \lesssim \frac{1.24\cdot 4.33^n}{\sqrt{n}}   &  \frac{0.065\cdot 12.65^n}{n\sqrt{n}} &\lesssim  E[C_{n,3,-3}]\lesssim \frac{0.11\cdot 12.65^n}{\sqrt{n}} \\
\end{align*}
First one sees that the nature of these growths is unaffected by a change of weights and/or values of $\theta$. This is not really surprising, since
the grammar is strongly connected and therefore always gives rise to generating functions whose singularities are of square-root type~\cite{Drmota97}.
However the exponential growth factor is strongly affected by these variations with practical consequences. For instance considering tRNAs ($n=80$)
and using our first order approximation gives a time of first collision of $\sim 4.7\,.\,10^{13}$ samples in the $(\theta=1,E=-1)$ model, while
only $\sim 93.55$ samples are required in the $(\theta=3,E=-3)$ model for the first collision to occur.

\subsection{Collisions and coverage}

\begin{figure}
  \begin{tabular}{cc}
  {\includegraphics[width=.45\textwidth]{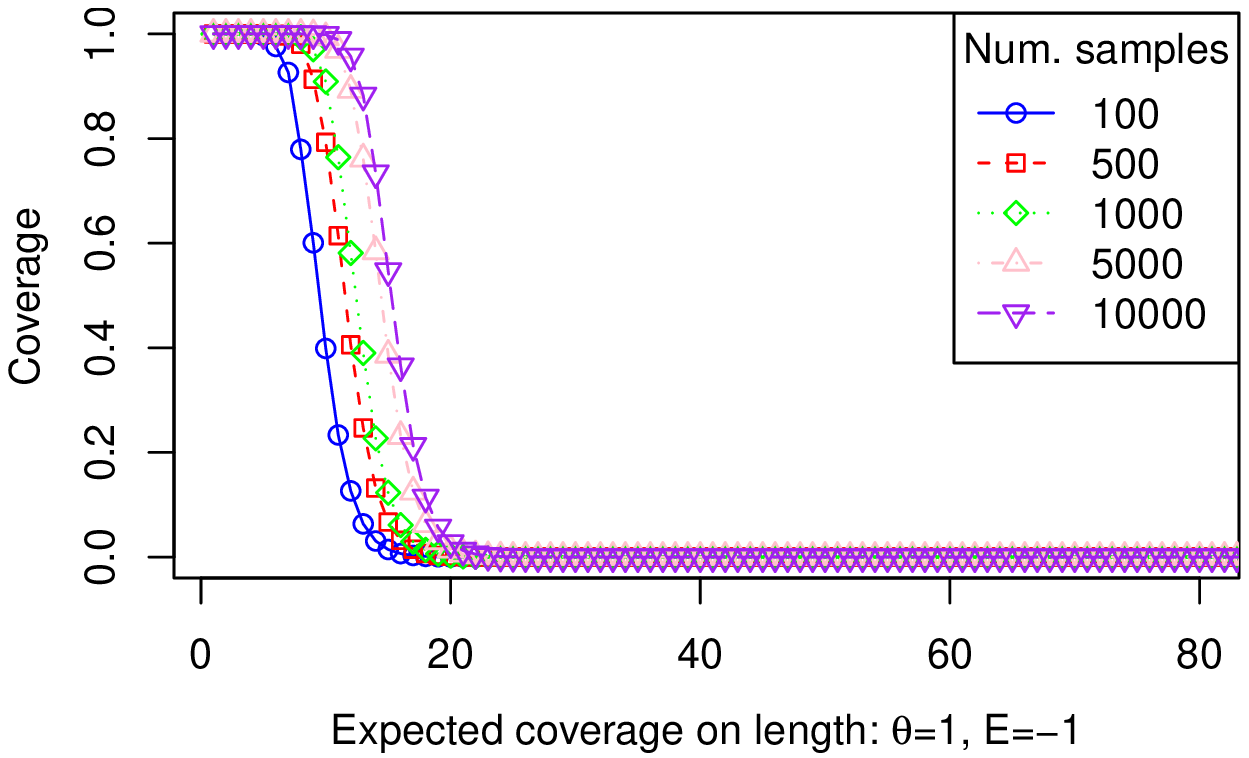}}&
  \includegraphics[width=.45\textwidth]{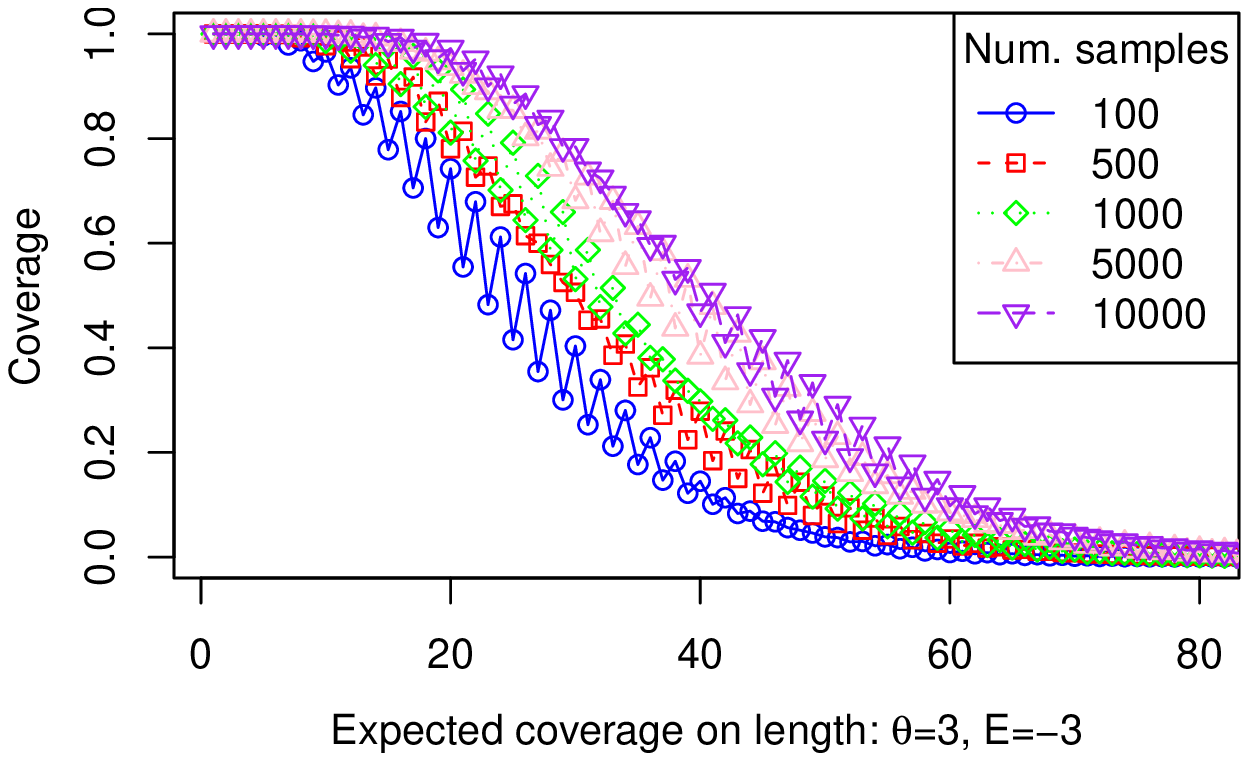}\\[0.6em]
  \includegraphics[width=.45\textwidth]{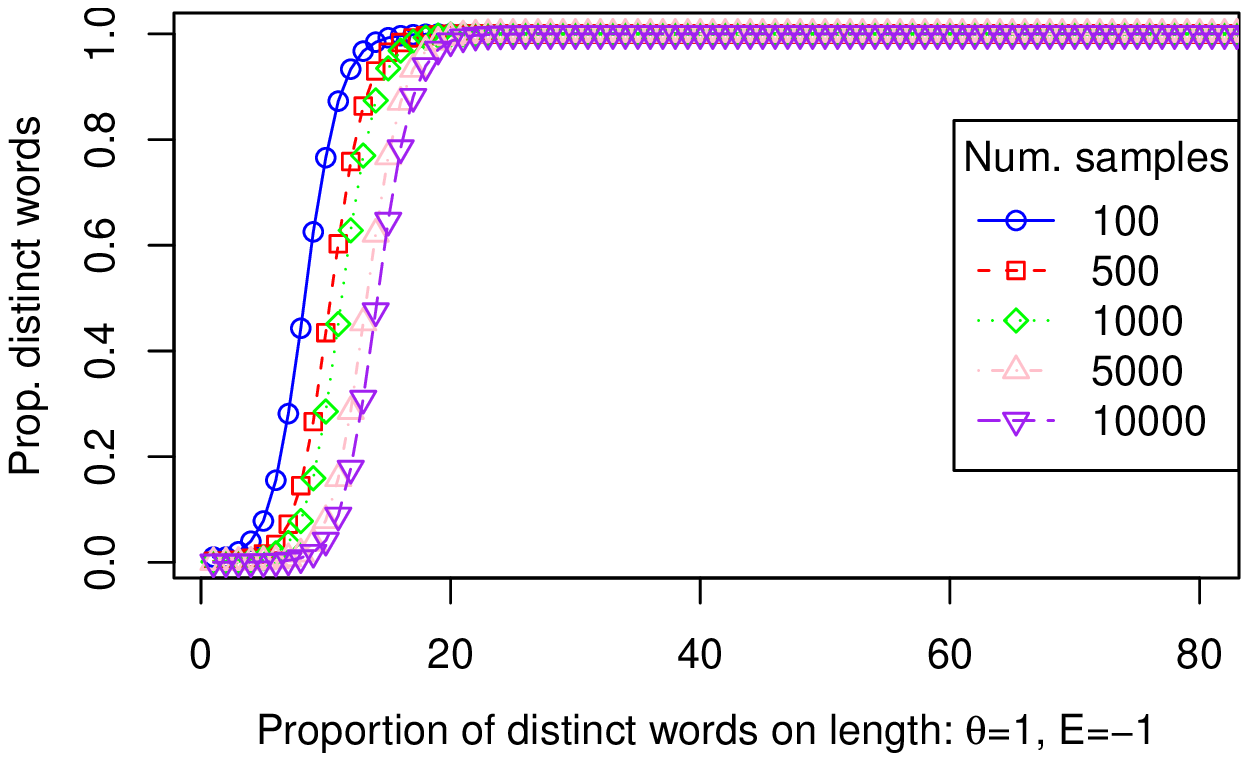}&
  \includegraphics[width=.45\textwidth]{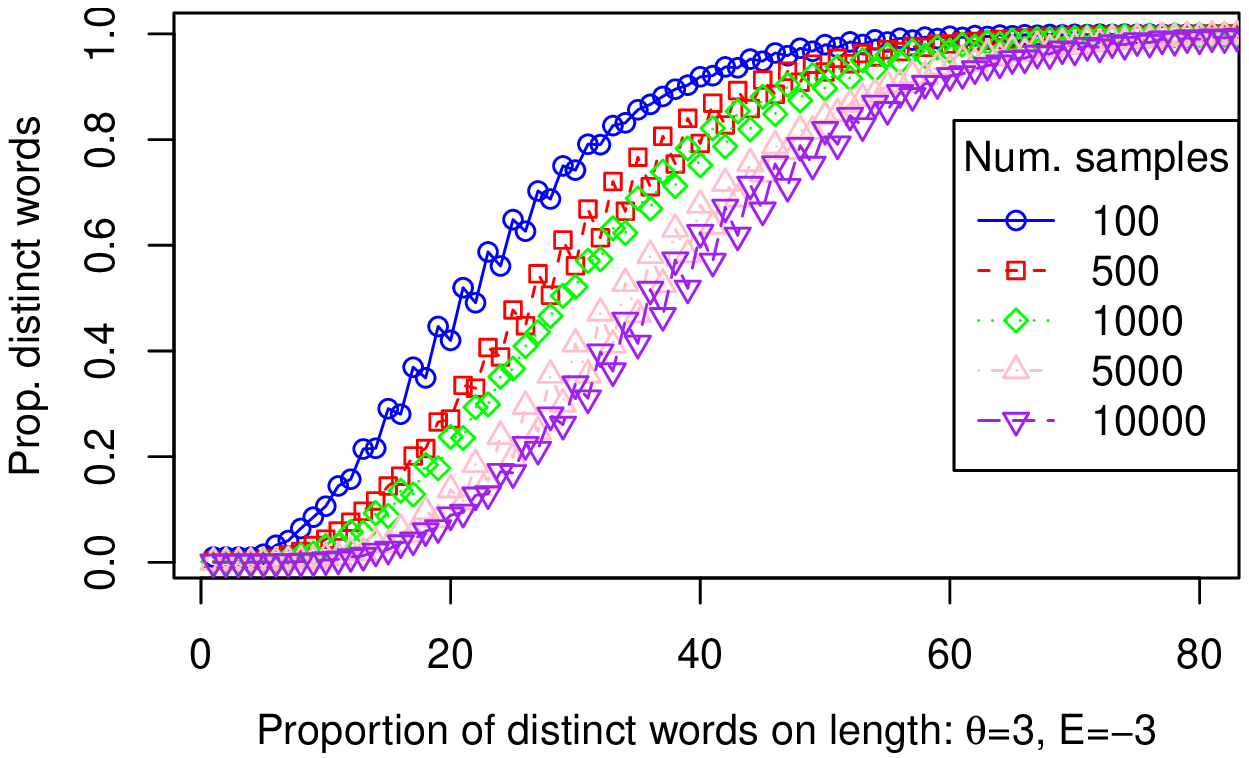}
  \end{tabular}
  \caption{Expected coverage (Top) and proportion of distinct words (Bottom) within sampled set of words of various length, considering different values
  for $\theta$ and $E$ the free-energy contribution of a base-pair.\label{fig:coverage}}
\end{figure}

Finally let us address the coverage and number of distinct samples obtained by a random generation scenario.
Remark that RNA secondary structures of length $n$ with $k$ plateaux are in bijection with Motzkin words of length $n-k\theta$ with $k$ peaks/plateaux,
where the bijection simply consists in removing the first $\theta$ horizontal steps of each plateau in every secondary structure.
Let us further remind that Dyck words with $k$ peaks and $2i$ letters are counted by the Naranaya numbers $\mathcal{N}(i,k)$,
and that Motzkin words are obtained by inserting some dots within a Dyck word.
It follows that the number $s_{n,k,i,\theta}$ of secondary structures of length $n$ featuring $i$ plateaux and $k\ge i$ base-pairs is such that
\begin{equation} s_{n,k,i,\theta} = \mathcal{N}(i,k)\binom{n-\theta k}{n-2i-\theta k} = \frac{1}{i}\binom{i}{k}\binom{i}{k-1}\binom{n-\theta k}{n-2i-\theta k}\end{equation}
Using the above formula, one can compute exactly in polynomial time the expected coverage from Theorem~\ref{th:weightedCoverage} and the proportion of distinct samples from Theorem~\ref{th:weightedDistinct},
and one obtains the results summarized in Figure~\ref{fig:coverage}.
Interestingly Figure~\ref{fig:coverage} shows that the inevitable decay of the coverage can be delayed by free-energies contributions of large absolute values.
For instance a sampled set of 1000 structures still achieves a 50$\%$ coverage for RNAs of length $<30$ for a free-energy contribution in
the $(\theta=3,E=-3)$ model while yielding a negligible coverage in the $(\theta=1,E=-1)$ model. This suggests that, for highly stable RNAs
(having low free-energy) of modest size, the 1000 structure criterion might be sufficient. Also a symmetry of the coverage and proportion
can be observed, although the amplitude of the oscillations for $\theta=3$ seem to have less of an impact on the proportion of distinct words than on their coverage.

\section{Conclusion and perspectives}\label{sec:conclusion}
In this article, we investigated the redundancy of random sets of words of context-free languages drawn with respect to a weighted distribution.
Using a random allocation model we derived exact and/or asymptotic equivalent forms for:
the expected numbers of generations prior to the first collision and full collection,
the average proportion of distinct words within a sampled set of k words and its cumulated probability.
Interestingly, the second moment of the probability distribution both appears in the asymptotic behaviors of the first
collision and the expected coverage.
We applied these theorems to analyze the output of a statistical sampling algorithm used to predict the functional
folding of RNA molecules. We showed that, although the time of first collision is exponential on the length of the RNA,
its exponential factor strongly depends on the free-energy contribution of base-pairs, and may still allow for frequent
collisions for RNAs of small -- yet relevant -- lengths.

Future directions for this work first include a better characterization of the full collection waiting time.
Namely we showed that, unsurprisingly, the waiting time is dominated by the overall (exponential) weight but
obtained lower and upper that are still separated by a $\Theta(n)$ factor. A possible direction for a tighter
bound resides in algebraic manipulations of Harmonic numbers coupled with additional
assumptions on the distribution of weights (i.e. distribution of symbols), for which local limit theorems are
known to hold under certain hypotheses. Also we may refine our analysis of RNA statistical sampling, using more sophisticated -- yet still context-free -- grammars
in order to accommodate more realistic models for the free-energy.

\section*{Acknowledgements}
The authors wish to thank the organizers of the GASCOM'08 conference in Bibbiena, Italy where the present collaboration started.
The present work was funded the \emph{Agence Nationale de la Recherche} through the BOOLE NT09$\_$432755 (DG) and the GAMMA 07-2$\_$195422 (YP) programs.

% ----------------------------------------------------------------
\bibliographystyle{amsplain}
\bibliography{biblio}
\newpage
\section*{Appendix}
\begin{proof}{(Theorem~\ref{th:birthday})}
From~\cite{FlaGarThi92}, the waiting time for the first birthday can be expressed as
\[
E(B) = \int_{0}^{+\infty} \lambda(t) e^{-t} dt,
\qquad
with
\qquad
\lambda(t) = \prod_{i=1}^m (1+ p_i t).
\]
Let us approximate this integral under the conditions of the theorem.
We cut the integral at $\tau$, and independently consider the part from~0 to~$\tau$,
which we expect will be dominating, and the part from~$\tau$ to~$+\infty$, which should give rise to a negligible contribution.

Let us first approximating the integral $\int_0^\tau \lambda(t) e^{-t} dt$. Consider $$\psi (t) = \log \lambda (t) - t = \sum_{i=1}^{m} \log(1+p_i t) -t.$$
Then $E(B) = \int_0^{+\infty} e^{\psi(t)} dt.$
Let us consider a positive real value $t<\tau$ such that any value $p_i t$ is uniformly bounded by some $A<1$, then
$\log (1+p_i t) = p_i t - p_i^2 t^/2 + \BigO{p_i^3 t^3}$, where the bound implied in the $\BigO{\cdot}$ term is uniform in~$p_i t$.
Summing over the whole distribution gives
\[
\psi(t) = - \left(\sum_i p_i^2\right) t^2/2 + \BBigO{\sum_i p_i^3 t^3}
\]
then $\psi(t) = - \alpha_2 t^2/2 + \BigO{\alpha_3 t^3}= - \alpha_2 t^2/2 + \BigO{\alpha^3 \tau^3}$.
Plugging this into the integral gives
\[
\int_0^\tau e^{- \alpha_2 t^2/2 + \BigO{\alpha^3 \tau^3}} dt = \int_0^\tau e^{- \alpha_2 t^2 /2} dt . (1+\BigO{\alpha_3 \tau^3}).
\]
This last integral is computed by a change of variable $u=t \sqrt{\alpha_2}$. Approximating with a Gaussian integral $\int_0 ^{+\infty} e^{- u^2/2} du = \sqrt{\pi/2}$ finally gives
\[
\int_0^\tau \lambda(t) e^{-t} dt = \sqrt{\frac{\pi}{2\alpha_2}} \,
\left( 1+ \BBigO{e^{- \tau^2 \alpha_2 /2}} + \BBigO{\alpha_3 \tau^3}
\right) .
\]
Of course, the validity of this expansion requires that
\begin{itemize}
\item The error terms in the above equation are $o(1)$: This follows from our assumptions on $\tau$, reminding that $\alpha_3 \tau^3\to 0$ (Condition~\ref{cond:thirdMoment}) and
$\alpha_2 \tau^2\to \infty$ (Condition~\ref{cond:secondMoment}).
\item Each of the terms $p_i t$ is uniformly bounded by some $A<1$: Since $p_m$ is the greatest probability, then it suffices that $p_m \tau \leq A <1$ (Condition~\ref{cond:smallLargestProba}).
\end{itemize}

Let us bound the value of the remainder $\int_\tau^{+\infty} \lambda(t) e^{-t} dt$.
We factor out the term $\lambda(\tau) e^{-\tau} = e^{\psi(\tau)}$, which we expect to be dominant.
The remaining term is
\begin{eqnarray*}
\int_\tau^{+\infty} \frac{ \lambda(t)}{\lambda(\tau)} e^{\tau -t} dt
&=&
\int_0^{+\infty} \frac{\lambda(\tau +s)}{\lambda(\tau)} e^{-s} ds
\\ &=&
\int_0^{+\infty} \prod_{i=1}^m \left( 1 + \frac{p_i}{1 + p_i \tau} \; s \right) \, e^{-s} ds
\\ &=&
\int_\tau^{+\infty} e^{\sum_{i=1}^m \log \left( 1 + \frac{p_i}{1 + p_i \tau} \; s \right) -s} ds
\end{eqnarray*}
Now, for any positive $x$, $\log(1+x) \leq x$ which gives a bound
\begin{eqnarray*}
\sum_{i=1}^m \log \left( 1 + \frac{p_i}{1 + p_i \tau} \; s \right) -s
&\leq&
\left( \sum_{i=1}^m \frac{p_i}{1 + p_i \tau} \; s \right) -s
\\ &\leq &
\sum_{i=1}^m \left( \frac{p_i}{1 + p_i \tau} - p_i \right) s = -s B (\tau),
\end{eqnarray*}
with $B (\tau) = \sum_{i=1}^m \left( \frac{- p_i}{1 + p_i \tau} + p_i \right) = \sum_i \frac{p_i^2 \tau}{1 + p_i \tau}.$
It follows that
\[
\int_0^{+\infty} \frac{\lambda(s+\tau)}{\lambda(\tau)} e^{-s} ds \leq
\int_0^{+\infty} e^{- B(\tau) s} ds = \frac{1}{B(\tau)}
\]
and finally
\[
\int_\tau^{+\infty} \lambda(t) e^{-t} dt \leq
\frac{\lambda(\tau) \, e^{-\tau}}{B(\tau)}.
\]

Let us consider the order of $B(\tau)$. We easily check that, for each $i$,
\[
0 < 1-p_i \tau < \frac{1}{1+p_i \tau} < 1,
\]
thus
\[
0 < p_i^2 \tau - p_i^3\tau^2 < \frac{p_i^2 \tau }{1+p_i \tau} < p_i^2 \tau,
\]
which gives bounds on $B(\tau)$ as
\[
\alpha_2 \tau - \alpha_3 \tau^2 < B(\tau) < \alpha_2 \tau.
\]

Finally, we can bound the error term. In order to conclude, we need to show that $$\lambda(\tau) e^{-\tau} / B(\tau) = o(1/\sqrt{\alpha_2}).$$
First rewrite the last condition as $e^{ - \alpha_2 \tau^2/2 + \BigO{\alpha_3 \tau^3}}= o(B(\tau) / \sqrt{\alpha_2})$,
taking advantage of $\lambda(\tau) =e^{\tau - \alpha^2 \tau^2/2 + O(\alpha_3 \tau^3)}$.
Assume that we have chosen $\tau$ such that $\alpha_2 \tau \rightarrow + \infty$
and $\alpha_3 \tau^3 \rightarrow 0$; then $B(\tau)$ has exact order $\alpha_2 \tau$ and the condition collapses to $e^{- \alpha_2 \tau^2/2 + O (\alpha_3 \tau^3) }= o (\tau \sqrt{\alpha_2})$, which is trivial.
\end{proof}

  \begin{proof}{(Theorem~\ref{prop:birthday})}
  Let us first remind that  the \Def{exponential order}~\cite{Flajolet2009} of a sequence $f_n$,
  is a simple exponential function $K^n$ such that $$\lim_{n\to \infty} \sup |f_n|^{1/n} = K.$$
  Following notations of the Flajolet/Sedgewick's book~\cite{Flajolet2009}, we make use of the \emph{bowtie} notation, and write $f_n\bowtie K^n$ if $f_n$
  has exponential order $K^n$.
  It is a classic result~\cite[Theorem IV.7]{Flajolet2009} that the dominant singularity $\rho$ of a generating
  function determines the exponential order of its coefficients $c_n$, namely through
  $ c_n \bowtie \rho^{-n}.$

  Since ${\rho_{\W}}^k < \rho_{\W^k}$ holds for any $k>1$ and $\W_0> {\bf 1}$ (Condition~\ref{cond:diversity}), then it follows that
  \begin{equation} s_{n,k}:=\sqrt[k]{\PF{\W_0^k}{n}} \bowtie \left(\sqrt[k]{\rho_{\W_0^k}}\right)^{-n} \text{\quad and \quad}
   \sqrt[k]{\rho_{\W_0^k}} > \rho_{\W_0}\label{eq:ineqMoments}\end{equation}
  for $\W_0$ any vector of weights strictly larger than $1$, and $\rho_{\W_0}$ the dominant singularity of $\WGF{L}{\W}$.

  This result generalizes to any pair $(a,b)\in\mathbb{R}^2$ of numbers such that $1<a<b$.
  Indeed, upon taking $\W_0 = \W^a$ and $k=b/a$ in the above equation, one has $s_{n,a}~\bowtie~\left(\sqrt[a]{\rho_{\W^a}}\right)^{-n}$,
  $s_{n,b} \bowtie \left(\sqrt[b]{\rho_{\W^b}}\right)^{-n}$, and it follows from Condition~\ref{cond:diversity} that
  \begin{equation}\sqrt[a]{\rho_{\W^a}}<\sqrt[b]{\rho_{\W^b}}.\end{equation}
  Consequently, for any $1<a<b$, $s_{n,a}$ grows exponentially faster than $s_{n,b}$, and one can use such a hierarchy to \emph{squeeze} $\tau_{n}^{-1}$
  between $\sqrt{\Mom_{2,n}}$ and $\sqrt[3]{\Mom_{3,n}}$.

  Namely let us consider $$\tau_{n}:= \frac{1}{\sqrt[k]{\Mom_{k,n}}}$$ for some $k\in \mathbb{Q}$ such that $2<k<3$. Then we have
  \begin{align*}
    \sqrt{\Mom_{2,n}} \cdot\tau_{n} = \sqrt{\frac{\PF{\W^2}{n}}{\PF{\W}{n}^2}}\sqrt[k]{\frac{\PF{\W}{n}^k}{\PF{\W^k}{n}}}
    = \frac{s_{n,2}}{s_{n,k}} \bowtie \left(\frac{\sqrt[k]{\rho_{\W^k}}}{\sqrt{\rho_{\W^2}}}\right)^{n}
  \end{align*}
  and it follows from $\sqrt{\rho_{\W^2}} < \sqrt[k]{\rho_{\W^k}}$ that
  $$  \lim_{n\to \infty} \sqrt{\Mom_{2,n}} \cdot\tau_{n} = +\infty$$
  and consequently Condition~\ref{cond:secondMoment} is satisfied by our candidate $\tau_{n}$.

  Reciprocally for Condition~\ref{cond:thirdMoment}, one has
    $$\sqrt[3]{\Mom_{3,n}} \cdot\tau_{n}  = \sqrt[3]{\frac{\PF{\W^3}{n}}{\PF{\W}{n}^3}}\sqrt[k]{\frac{\PF{\W}{n}^k}{\PF{\W^k}{n}}}
    = \frac{s_{n,3}}{s_{n,k}} \bowtie \left(\frac{\sqrt[k]{\rho_{\W^k}}}{\sqrt[3]{\rho_{\W^3}}}\right)^{n}$$
  and, since $\sqrt[3]{\rho_{\W^3}} > \sqrt[k]{\rho_{\W^k}}$, then
  $$  \lim_{n\to \infty} \sqrt[3]{\Mom_{3,n}} \cdot\tau_{n} = 0.$$

  Condition~\ref{cond:smallLargestProba} is also satisfied by $\tau_{n}$ upon observing that
   \begin{align*} \PMax{\W}{n} \cdot\tau_n & =  \frac{\WMax{\W}{n}}{\PF{\W}{n}}\sqrt[k]{\frac{\PF{\W}{n}^k}{\PF{\W^k}{n}}} = \frac{\WMax{\W}{n}}{s_{n,k}}\end{align*}
   where $\WMax{\W}{n}$ is the  weight of the heaviest (i.e. most probable) word $w^{\triangle}\in \Lang_n$. This word is also contributing to $\PF{\W^k}{n} = \sum_{w\in\Lang_n}\WF{w}^k$
   and therefore
   $$ s_{n,k} = \sqrt[k]{\PF{\W^k}{n}} = \sqrt[k]{{\WMax{\W}{n}}^k+\sum_{\substack{w\in\Lang_n\\ w\neq w^{\triangle}}}\WF{w}^k} > \WMax{\W}{n}$$
   which suffices to prove that Condition~\ref{cond:smallLargestProba} is satisfied.
   Consequently, the preconditions of Theorem~\ref{th:birthday} are satisfied by any weighted distribution.
  \end{proof}

\end{document}